\documentclass[a4paper]{amsart}

\usepackage{amssymb}
\usepackage{hyperref}
\usepackage{xcolor}
\usepackage{enumerate}

\usepackage{tikz}
\usetikzlibrary{arrows,automata,backgrounds,calc,positioning}

\tikzstyle{state} = [circle,draw, inner sep = 0, minimum size = 15pt]

\newtheorem{theorem}{Theorem}
\numberwithin{theorem}{section}
\newtheorem{definition}[theorem]{Definition}

\newtheorem{lemma}[theorem]{Lemma}
\newtheorem{corollary}[theorem]{Corollary}
\newtheorem{proposition}[theorem]{Proposition}
\newtheorem{remark}[theorem]{Remark}

\newcommand{\D}{\mathcal{D}}
\newcommand{\R}{\mathcal{R}}

\newcommand{\FT}{\mathrm{FT}}
\newcommand{\Runs}{\mathrm{Runs}}
\newcommand{\rev}{\mathsf{R}}

\newcommand{\eps}{\epsilon}

\newcommand{\last}{\mathrm{last}}
\newcommand{\wnd}{\mathrm{wnd}}
\newcommand{\mwl}{\mathrm{mwl}}

\newcommand{\Reg}{\mathbf{Reg}}
\newcommand{\Len}{\mathbf{Len}}
\newcommand{\LI}{\mathbf{LI}}
\newcommand{\RI}{\mathbf{RI}}
\newcommand{\ST}{\mathbf{ST}}
\newcommand{\PT}{\mathbf{PT}}
\newcommand{\PF}{\mathbf{PF}}
\newcommand{\SF}{\mathbf{SF}}
\newcommand{\LB}{\mathbf{LB}}
\newcommand{\RB}{\mathbf{RB}}

\newcommand{\SWA}{SWA}
\newcommand{\SWAs}{SWAs}

\title{Randomized sliding window algorithms for regular languages}

\author[M.~Ganardi]{Moses Ganardi}
\author[D.~Hucke]{Danny Hucke}
\author[M.~Lohrey]{Markus Lohrey}
\address{Universit\"at Siegen, Germany\\
  \texttt{\{ganardi,hucke,lohrey\}@eti.uni-siegen.de}}

\begin{document}

\maketitle

\begin{abstract}
         A sliding window algorithm receives a stream of symbols and has to output
	at each time instant a certain value which only depends on the last $n$ symbols.
	If the algorithm is randomized, then at each time instant it produces an incorrect output
	with probability at most $\eps$, which is a constant error bound.
	This work proposes a more relaxed definition of correctness which is parameterized by
	the error bound $\eps$ and the failure ratio $\phi$:
	A randomized sliding window algorithm is required to err with probability at most $\eps$
	at a portion of $1-\phi$ of all time instants of an input stream.
	
	This work continues the investigation of sliding window algorithms for regular languages.
	In previous works a trichotomy theorem was shown for deterministic algorithms:
	the optimal space complexity is either constant, logarithmic or linear
	in the window size.
	The main results of this paper concerns three natural settings
	(randomized algorithms with failure ratio zero and randomized/deterministic algorithms with bounded failure ratio)
	and provide natural language theoretic characterizations of the space complexity classes.
\end{abstract}

\section{Introduction}

{\em Sliding window algorithms} process an input sequence 
$a_1 a_2 \cdots a_m$ from left to right and have at time $t$ only direct access to the current symbol $a_t$. 
Moreover, at each time instant $t$ the algorithm is required to compute a value that depends on
the last $n$ symbols.
The value $n$ is called the {\em window size} and the last $n$ symbols form the {\em active window} at time $t$.
	
In many streaming applications, data items are outdated
after a certain time. The sliding window model is a simple way to model this. 
A typical application for sliding window algorithms is the analysis of a
time series as it may arise in medical monitoring, web tracking, or financial monitoring.
A detailed introduction into the sliding window model can be found in \cite[Chapter~8]{Aggarwal07}.
	
A general goal in the area of sliding window algorithms is to avoid the explicit storage of the window content, and, instead, to work in
considerably smaller space, e.g. polylogarithmic space with respect to the window length.
In the seminal paper of Datar et al.~\cite{DatarGIM02}, where the sliding window model was introduced,
the authors prove that the number of $1$'s in a $0/1$-sliding window of size $n$
can be maintained in space $\mathcal{O}(\frac{1}{\eps} \cdot \log^2 n)$ if one allows a multiplicative error of $1\pm \eps$.
Other algorithmic problems that were addressed in the extensive literature on sliding window streams
include the computation of statistical  data
(e.g. computation of  the variance and $k$-median \cite{BabcockDMO03}, and quantiles \cite{ArasuM04}),
optimal sampling from sliding windows \cite{BravermanOZ12},
the membership problem of regular languages \cite{GHL16}, computation of edit distances \cite{CLLPTZ11},
database querying (e.g. processing of join queries over sliding
windows \cite{GolabO03}) and graph problems (e.g. checking for connectivity and computation of matchings, spanners, and
minimum spanning trees~\cite{CrouchMS13}).
The reader can find further references in the surveys  \cite[Chapter~8]{Aggarwal07} and \cite{Braverman16}.
	
In our recent papers \cite{GHKLM18,GHL16} we studied the space complexity of deterministic sliding window algorithms
for regular languages.  Such an algorithm returns at every time instant $1$ (resp., $0$) if the active window belongs to a fixed regular language $L$. In \cite{GHL16} we proved that for every regular language $L$
the optimal space bound for a sliding window algorithm for $L$ is either constant, logarithmic or linear in the window size.
In \cite{GHKLM18} we also gave several characterizations
for these space classes: the class of regular languages that have a sliding window algorithm with space complexity
$\mathcal{O}(\log n)$ is the Boolean closure of all regular length languages 
and all regular left ideals, and the class of regular languages that have a sliding window algorithm with space complexity
$\mathcal{O}(1)$ is the Boolean closure of all regular length languages 
and all suffix-testable languages. The definitions of these language classes can be found in Section~\ref{sec-classes}.

In this paper, we extend our results from  \cite{GHKLM18,GHL16} to randomized sliding window algorithms, i.e.,
Monte-Carlo sliding window algorithms that can err with a small probability.
Since a sliding window algorithm produces an output after each input symbol,
there are different ways to interpret this correctness condition.
The maybe most natural answer to this question is
to require that after reading an arbitrary input word $a_1 a_2 \cdots a_m$, the algorithm gives an incorrect answer
to the question whether $a_{m-n+1} \cdots a_m \in L$ ($n$ is the window size, $L$ is the regular language under consideration) 
with probability at most $\eps$, where $\eps$ is a fixed constant strictly smaller than $1/2$. This ensures that
for every input stream and every time instant, one can be sure to get a correct answer with probability at least $1-\eps$.
This is certainly a natural requirement, but one may argue that it is not crucial if the algorithm produces a wrong answer with higher probability on a small proportion
of all time instants. This leads us to the more general definition of randomized sliding window algorithms: Fix two parameters
$\eps$ (the error probability) and $\phi$ (the failure ratio) with $0 \leq \eps < 1/2$ and $0 \leq \phi \leq 1$.
We say that a randomized sliding window algorithm (for a certain language $L$ and a window size $n$)
is $(\eps,\phi)$-correct if for every input stream, the portion of all time instants where the algorithm gives a wrong answer with probability larger than 
$\eps$ is bounded by $\phi$.
Using a standard probability amplification argument, one can
show that any error probability $\eps < 1/2$ can be reduced to any constant $\eps' > 0$ without increasing
the failure ratio. Thereby, the space only increases by a multiplicative constant (that depends on $\eps$ and $\eps'$).


Using the definition of $(\eps,\phi)$-correctness, this paper analyzes the space complexity of regular languages in the following natural cases:
\begin{enumerate}[(i)]
\item $(\eps,0)$-correct randomized sliding window algorithms, where $0 < \eps < 1/2$, i.e., at every time instant the randomized algorithm gives a correct answer with probability at least $1-\eps$.
\item $(\eps,\phi)$-correct randomized sliding window algorithms, where $0 < \eps < 1/2$ and $0 < \phi < 1$, i.e., there is a portion $\phi$ of time instants where the answer might be wrong with probability larger than $\eps$. Here we are interested in randomized algorithms where $\phi$ is  an arbitrarily small constant larger than zero.
\item $(0,\phi)$-correct deterministic sliding window algorithms, where $0 < \phi < 1$, i.e., the deterministic algorithm has the 
property that for every input stream only a $\phi$-portion of the produced outputs is wrong. Again we are interested in the case where $\phi$ is an arbitrarily small constant larger than zero.
\end{enumerate}
For each of these three settings we prove a main result that completely characterizes the space complexity of regular languages.
Figure~\ref{fig-survey} in Section~\ref{sec-main} shows the resulting space classes for each setting; the left column
shows the three classes for the deterministic setting studied in \cite{GHKLM18,GHL16}.
Below, we give a more detailed explanation of our main results.

Theorem~\ref{thm:quatrochotomy} deals with setting (i), i.e., $(\eps,0)$-correct randomized sliding window algorithms, where $0 \leq \eps < 1/2$. Note that the deterministic case
$\eps=0$ is considered in \cite{GHKLM18,GHL16}. For $\eps > 0$, Theorem~\ref{thm:quatrochotomy}  states
a space quatrochotomy (in contrast
to the space trichotomy for $\eps=0$): for every regular language the optimal space is either constant, doubly logarithmic, logarithmic or linear. 
For each of the four space classes we present a natural language theoretical characterization. 
It turns out that when going from the deterministic to the randomized setting, then the novelty is an improvement for some regular languages from logarithmic to doubly logarithmic space. 
The new doubly logarithmic space class is the Boolean closure
of regular suffix-free languages, suffix-testable languages, and regular length languages.

Let us mention that Tesson and Th\'erien \cite{TessonT05} proved a quatrochotomy
(resp., trichotomy) result for the randomized (resp., deterministic) communication complexity
for regular languages. This results resembles Theorem~\ref{thm:quatrochotomy} (resp., the 
trichotomy result proved in  \cite{GHL16}), but the language classes that appear in  \cite{TessonT05}
are different from the classes in our results and we do not see a deeper connection.

Theorem~\ref{thm:dichotomy} characterizes the space complexity of regular languages with respect to 
$(\eps,\phi)$-correct  randomized sliding window algorithms, where $0 \le \eps < 1/2$ and $0 < \phi < 1$ (setting (ii)).
We show that there is a subclass $\mathcal{C}$ of regular languages such that for 
every regular language $L \in \mathcal{C}$ and every $0 < \eps < 1/2$, $0 < \phi < 1$, there exists
a constant-space $(\eps,\phi)$-correct  randomized sliding window algorithm for $L$. On the other hand
for every regular language $L$ that does not belong to $\mathcal{C}$, there exists a threshold $\phi_0$ such that
there is no $(\eps,\phi)$-correct  randomized sliding window algorithm for $L$ that uses space $o(n)$ with $0 \leq \eps < 1/2$ and
$\phi \leq \phi_0$. The class $\mathcal{C}$ is characterized as the Boolean closure of regular left ideals, regular prefix-free languages, and regular length languages.

Theorem~\ref{thm:trichotomy} deals with setting (iii), i.e., the space complexity of regular languages with respect to 
$(0,\phi)$-correct deterministic sliding window algorithms. Similarly to Theorem~\ref{thm:dichotomy}, we are interested in the case where $0<\phi<1$ can be arbitrarily small.
We show that for every regular language the optimal space in this setting is either 
constant, logarithmic or linear.
The class of regular languages that need logarithmic space is the class $\mathcal{C}$ described above. The regular languages that need constant space is the Boolean closure of regular length languages, regular prefix-free languages, regular suffix-free languages and languages of the form $\Sigma^*L$, where $L$ is regular, prefix-free and suffix-free.

In Theorems~\ref{thm:quatrochotomy} and \ref{thm:dichotomy}, we consider randomized sliding window algorithms with 
a two sided error (analogously to the 
complexity class {\sf BPP}). Randomized sliding window algorithms with a one-sided error (analogously to the 
class {\sf RP})  can be motivated by applications, where all ``yes'' outputs have to be correct, but 
a small probability for a false negative answer is acceptable. In Section~\ref{sec-one-sided} we prove that for every
regular language the optimal space bound with respect to randomized sliding window algorithms with one-sided
error coincides (up to constant factors) with the optimal space bound in the deterministic setting
\cite{GHKLM18,GHL16} (which was discussed in the introduction). In other words: randomized sliding window algorithms
with a one-sided error can be derandomized.

Finally, in Section~\ref{sec-strict} we consider a more restricted notion of correctness for randomized sliding
window algorithms: Let us say that a randomized sliding window algorithm is strictly $\eps$-correct if for every
input stream $w$, the probability that the algorithm gives an incorrect output at some time instant is 
at most $\eps$. In other words: with probability $1-\eps$ all outputs produced while running
over $w$ are correct. This correctness notion is used for instance in \cite{BEFK16,DatarGIM02}.
Using a probabilistic argument we show that every randomized strictly $\eps$-correct sliding window algorithm
can be transformed into a  deterministic correct sliding window algorithm without increasing space. We show
this result not only for regular languages, but for all approximation problems, where an approximation problem
is formalized as a relation $\Pi \subseteq \Sigma^* \times \Omega$ where
$\Sigma$ is a finite alphabet and $\Omega$ is a (possibly infinite) set of output values.
If the active window is $w$, then the output $a$ of the algorithm is considered to be correct
if $(w,a) \in \Pi$.

\section{Preliminaries} \label{prel}

For integers $i, j \in \mathbb{N}$ let $[i,j] = \{ k \in \mathbb{N} \colon i \leq k \leq j \}$.
The set of all words over a finite alphabet $\Sigma$ is denoted by $\Sigma^*$.
The empty word is denoted by $\varepsilon$ whereas error probabilities are denoted by the lunate epsilon $\eps$.
The sets of words over $\Sigma$ of length exactly, at most and at least $n$ are denoted by
$\Sigma^n$, $\Sigma^{\le n}$ and $\Sigma^{\ge n}$, respectively.
Consider a word $w = a_1 a_2 \cdots a_m$.
The {\em reversal} of $w$ is defined as $w^\rev = a_m \cdots a_2 a_1$, and for a 
language $L$ we set $L^\rev = \{ w^\rev \colon w \in L \}$. For a non-empty
interval $[i,j] \subseteq [1,m]$ we define $w[i,j] = a_i a_{i+1} \cdots a_j$.
If $i > j$ we set $w[i,j] = \varepsilon$.
A {\em prefix} of $w$ is a word of the form $w[1,i]$ for some $0 \le i \le m$;
a {\em suffix} of $w$ is a word of the form $w[i,m]$ for some $1 \le i \le m+1$.

A language $L \subseteq \Sigma^*$ is {\em prefix-free} (resp., {\em suffix-free})  if there are no two words $x,y \in L$
with $x \neq y$ and $x$ is a prefix (resp., suffix) of $y$.
A language is {\em bifix-free} if it is both prefix- and suffix-free.

\subsection{Automata and regular languages} \label{sec-classes}

For general background in automata theory see \cite{HoUl79}.
A {\em deterministic finite automaton} (DFA) $A = (Q,\Sigma,q_0,\delta,F)$
consists of a finite set of states $Q$, a finite alphabet $\Sigma$,
an initial state $q_0 \in Q$, a transition function $\delta \colon Q \times \Sigma \to Q$
and a set of final states $F \subseteq Q$.
We inductively extend $\delta$ to a function $\delta \colon Q \times \Sigma^* \to Q$ as usual:
$\delta(q,\varepsilon) = q$ and $\delta(q,xa) = \delta(\delta(q,x),a)$ for all $q \in Q$, $x \in \Sigma^*$, $a \in \Sigma$.
If $P \subseteq Q$ is a set of states, then $L(A,P) = \{ w \in \Sigma^* \colon \delta(q_0,w) \in P \}$.
The language accepted by $A$ is $L(A) = L(A,F)$. 
A language is {\em regular} if it is accepted by a DFA.

Classes of languages are denoted by boldfaces letters.
In this paper we will deal with the following language classes:
\begin{itemize}
\item $\Reg$: the class of all {\em regular languages}.
\item $\Len$: the class of {\em regular length languages}, i.e., regular languages $L$ such that for all $n\in\mathbb{N}$
we have $\Sigma^n \subseteq L$ or $\Sigma^n \cap L=\emptyset$.\item $\LI$: the class of {\em regular left ideals}, i.e., languages of the form $\Sigma^* L$ where $L$ is regular.
\item $\RI$: the class of {\em regular right ideals}, i.e., languages of the form $L \Sigma^*$ where $L$ is regular.
We have  $\RI = \{ L^\rev \colon L \in \LI\}$.
\item $\ST$: the class of {\em suffix testable languages}, i.e., languages $L$ which are $k$-suffix testable for some $k \ge 0$.
A language is {\em $k$-suffix testable} if it is a Boolean combination of languages $\Sigma^* w$ where $|w| \le k$.
\item $\PT$: the class of {\em prefix testable languages} is $\{ L^\rev \colon L \in \ST\}$.
\item $\SF$: the class of {\em regular suffix-free languages}
\item $\PF$: the class of {\em regular prefix-free languages}
\item $\LB$: the class of {\em left ideals generated by regular bifix-free languages}, i.e.,
languages of the form $\Sigma^* L$ where $L \in \PF \cap \SF$.
\item $\RB$: the class of {\em right ideals generated by regular bifix-free languages} is $\{ L^\rev \colon L \in \LB \}$.
\end{itemize}
It is easy to see that every finite language is prefix testable and suffix testable. Moreover, prefix testable
and suffix testable languages are regular.

\begin{lemma}
	We have the relations $\langle \ST \rangle \subseteq \langle \LB \rangle \subseteq \langle \LI \rangle$
	and $\langle \SF \rangle \subseteq \langle \LI \rangle$.
\end{lemma}

\begin{proof}
	Notice that every language of the form $\Sigma^* w$ contained in $\LB$
	because $\{w\}$ is bifix-free.
	The containment $\LB \subseteq \LI$ is clear.
	Finally, if $L$ is suffix-free, then one can easily see that $L = \Sigma^*L \setminus (\Sigma^* (\Sigma L))$,
	which implies $\SF \subseteq \langle \LI \rangle$.
\end{proof}

We remark that one could have defined the classes $\LI$ and $\LB$ also differently
(and similarly, $\RI$ and $\RB$).
For any language $L$ we have $\Sigma^* L = \Sigma^* \text{min}(L)$ where $\text{min}(L)$
is the set of minimal words in $L$ with respect to the suffix relation.
Since $\text{min}(\cdot)$ preserves regularity and $\text{min}(L)$ is suffix-free, $\LI$ is also the class
of languages of the form $\Sigma^* L$ where $L$ is regular and suffix-free.
Similarly, $\LB$ is the class of languages of the form $\Sigma^* L$ where $L$ is regular and prefix-free.
Since our proofs related to $\LB$ in fact yield decomposition of the form $\Sigma^* L$
where $L$ is bifix-free, we decided to define $\LB$ as above.

A class of languages $\mathbf{A}$ over $\Sigma$ is {\em Boolean closed}
if $K,L \in \mathbf{A}$ implies $\Sigma^* \setminus L \in \mathbf{A}$ and $K \cup L \in \mathbf{A}$. 
If $\mathbf{A}_1, \dots, \mathbf{A}_n$ are classes of languages over some alphabet $\Sigma$,
then $\langle \mathbf{A}_1, \dots, \mathbf{A}_n \rangle$
denotes the Boolean closure of $\bigcup_{i=1}^n\mathbf{A}_i$,
i.e., the smallest Boolean closed class which contains $\bigcup_{i=1}^n\mathbf{A}_i$.

\subsection{Approximation problems}

An {\em approximation problem} is a relation $\Pi \subseteq \Sigma^* \times \Omega$ where
$\Sigma$ is a finite alphabet and $\Omega$ is a (possibly infinite) set of output values.
For a given input word $w \in \Sigma^*$ the set of admissible outputs
is $\{ a \in \Omega \colon (w,a) \in \Pi\}$.
Typical examples include:
\begin{itemize}
\item exact computation problems $\Pi \colon \Sigma^* \to \Omega$. Typical examples are
the mapping $c_1 \colon \{0,1\}^* \to \mathbb{N}$ with $c_1(w) = $ ``number of $1$'s in $w$",
or the characteristic function $\chi_L \colon \Sigma^* \to \{0,1\}$ of a language $L \subseteq \Sigma^*$.
\item approximate statistics $\Pi \subseteq \Sigma^* \times \mathbb{N}$. A typical example
would be the set of all pairs $(w,k)$ such that $(1-\eps) \cdot c_1(w) \leq k \leq (1+\eps) \cdot c_1(w)$
for some (small) $\eps > 0$.
\end{itemize}
In this paper we will focus on language membership problems,
where we identify a language $L \subseteq \Sigma^*$ with its characteristic function $\chi_L \colon \Sigma^* \to \{0,1\}$.
Only in Section~\ref{sec-strict} we will talk about general approximation problems.

\subsection{Probabilistic automata with output}

In the following we will introduce probabilistic automata \cite{Paz71,Rabin63} as a model of randomized streaming algorithms
which produce an output after each input symbol.
A {\em probabilistic automaton} $R = (Q,\Sigma,\iota,\rho,\omega)$
consists of a (possibly infinite) set of states $Q$, an alphabet $\Sigma$,
an initial state distribution $\iota \colon Q \to \{ p \in \mathbb{R} \colon 0 \leq p \leq 1\}$,
a transition probability function $\rho \colon Q \times \Sigma \times Q \to \{ p \in \mathbb{R} \colon 0 \leq p \leq 1\}$
and an output function $\omega \colon Q \to \Omega$ such that
\begin{enumerate}
\item $\sum_{q \in Q} \iota(q) = 1$,
\item $\sum_{q \in Q} \rho(p,a,q) = 1$ for all $p \in Q$, $a \in \Sigma$.
\end{enumerate}
If $\iota$ and $\rho$ map into $\{0,1\}$, then $R$ is a {\em deterministic automaton}.
If $\Omega = \{0,1\}$ we specify the set $F \subseteq Q$ of final states instead.
A {\em run} on a word $a_1 \cdots a_m \in \Sigma^*$ in $R$ is a sequence $\pi = (q_0,a_1,q_1,a_2,\dots,a_m,q_m)$
where $q_0, \dots, q_m \in Q$ and $\rho(q_{i-1},a_i,q_i) > 0$ for all $1 \le i \le m$.
We write runs in the usual way
\[
	\pi : q_0 \xrightarrow{a_1} q_1 \xrightarrow{a_2} \cdots \xrightarrow{a_m} q_m
\]
or also omit the intermediate states: $\pi: q_0 \xrightarrow{a_1 \cdots a_m} q_m$.
We extend $\rho$ to runs in the natural way:
If $\pi : q_0 \xrightarrow{a_1} q_1 \xrightarrow{a_2} \cdots \xrightarrow{a_m} q_m$ is a run in $R$
then $\rho(\pi) = \prod_{i=1}^n \rho(q_{i-1},a_i,q_i)$.
Furthermore we define $\rho_\iota(\pi) = \iota(q_0) \cdot \rho(\pi)$.
We denote by $\Runs(R,w)$ the set of all runs on $w$ in $R$
and denote by $\Runs(R,q,w)$ those runs on $w$ that start in $q \in Q$.
Usually we simply write $\Runs(w)$ and $\Runs(q,w)$.
Notice that for each $w \in \Sigma^*$ the function $\rho_\iota$ is a probability distribution on $\Runs(w)$
and for each $q \in Q$ the restriction of $\rho$ to $\Runs(q,w)$
is a probability distribution on $\Runs(q,w)$.

\section{Randomized streaming and sliding window algorithms}


A {\em randomized streaming algorithm} $(R,\mathrm{enc})$
consists of a probabilistic automaton $R= (Q,\Sigma,\iota,\rho,\omega)$ as 
above and an injective function $\mathrm{enc} \colon Q \to \{0,1\}^*$.
Usually, we will only refer to the underlying automaton $R$.
If $R$ is deterministic, we speak of a {\em deterministic streaming algorithm}.
The maximum number of bits stored in a run
$\pi: q_0 \xrightarrow{a_1} q_1 \xrightarrow{a_2} \cdots \xrightarrow{a_m} q_m$ is denoted by $\mathrm{space}(R,\pi)$, i.e.,
\[
	\mathrm{space}(R,\pi) = \max \{ |\mathrm{enc}(q_i)| : 0 \le i \le m\}.
\]
We are interested in two measures of space complexity:
\begin{itemize}
\item worst case space complexity: 
\[ \mathrm{space}(R,w) = \max \{ \mathrm{space}(R,\pi) : \pi \in \Runs(w), \rho_\iota(\pi) > 0 \}, \]
\item expected space complexity: 
\[ \mathrm{space}_\varnothing(R,w) = \sum_{\pi \in \Runs(w)} \rho_\iota(\pi) \cdot \mathrm{space}(R,\pi)
\]
\end{itemize}
Let $R= (Q,\Sigma,\iota,\rho,\omega)$ be a randomized streaming algorithm, let $\Pi \subseteq \Sigma^* \times \Omega$
be an approximation problem
and let $w = a_1a_2 \cdots a_m \in \Sigma^*$ be an input stream.
\begin{itemize}
\item A run $\pi: q_0 \xrightarrow{w} q_m$ is {\em correct for $\Pi$}
if $(w,\omega(q_m)) \in \Pi$.
The {\em error probability} of $R$ on $w$ for $\Pi$ is
\[
	\eps(R,w,\Pi) = \sum_{\pi \in N} \rho_\iota(\pi),
\]
where $N = \{ \pi \in \Runs(w) \colon \text{$\pi$ is not correct for $\Pi$} \}$.
\item Given an error bound $0 \le \eps \le 1$ the algorithm $R$ {\em fails} at time instant $t \in [0,m]$
if $\eps(R,a_1a_2 \cdots a_t,\Pi) > \eps$.
The set of time instants at which $R$ fails on $w$ is denoted by
\[
	\FT(R,w,\Pi,\eps) = \{ t \in [0,m] \colon \eps(R,a_1a_2 \cdots a_t,\Pi) > \eps \}.
\]
The {\em failure ratio} of $R$ on $w$ is defined as
\[
	\phi(R,w,\Pi,\eps) = \frac{1}{m+1} |\FT(R,w,\Pi,\eps)|.
\]
\item A run $\pi: q_0 \xrightarrow{a_1} q_1 \xrightarrow{a_2} \cdots q_{m-1} \xrightarrow{a_m} q_m$ is {\em strictly correct for $\Pi$}
if $(a_1 \cdots a_t,\omega(q_t)) \in \Pi$ for all $0 \le t \le m$.
The {\em strict error probability} of $R$ on $w$ for $\Pi$ is
\[
	\eps_*(R,w,\Pi) = \sum_{\pi \in N_*} \rho_\iota(\pi),
\]
where $N_* = \{ \pi \in \Runs(w) \colon \text{$\pi$ is not strictly correct for $\Pi$} \}$.
\end{itemize}

\subsection{Basic properties of randomized streaming algorithms}

Before we specialize streaming algorithms to sliding window algorithms, we state two simple general properties
of randomized streaming algorithms that will be used implicitly throughout this paper.
The following lemma states that the error probability can be reduced to any non-zero constant.
Thereby the space only increases by a constant factor. 

\begin{lemma}[probability amplification]
	\label{lem:amplification}
	Given a language $L \subseteq \Sigma^*$, a randomized streaming algorithm $R$ and
	error bounds $0 < \eps' < \eps < \frac{1}{2}$, one can construct a randomized streaming
	algorithm $R'$ such that
	$\phi(R',w,L,\eps') \le \phi(R,w,L,\eps)$
	and $\mathrm{space}(R',w) \le \ln(\frac{1}{\eps'}) \cdot \frac{1}{\mathrm{poly}(\eps)} \cdot \mathrm{space}(R,w)$ for all $w \in \Sigma^*$.
\end{lemma}

\begin{proof}
	The algorithm $R$ simulates $k$ (which will be fixed later) instances of the algorithm in parallel with independent random bits
	and outputs the majority vote.
	Consider a stream $w \in \Sigma^*$ and a prefix $v$ of $w$ such that $\eps(R,v,L) \le \eps$.
	By the Chernoff bound (see e.g. \cite[Chapter~4]{MiUp17}) we know
	\[
		\eps(R',v,L) \le \exp\left(-\frac{k(\frac{1}{2}-\eps)^2}{2(1-\eps)}\right).
	\]
	If we choose $k \ge \ln\big(\frac{1}{\eps'}\big) \cdot \frac{2(1-\eps)}{(\frac{1}{2}-\eps)^2}$
	we get $\eps(R',v,L) \le \eps'$.
	This implies $\phi(R',w,L,\eps') \le \phi(R,w,L,\eps)$.
\end{proof}

Let $\Pi_1, \dots, \Pi_k \subseteq \Sigma^* \times \Omega$ be approximation problems 
and $\tau \colon \Omega^k \to \Omega$ be a mapping.
Then $\tau(\Pi_1, \dots, \Pi_k)$ denotes the approximation problem
\[
	\{ (x,\tau(y_1, \dots, y_k)) \colon x \in \Sigma^*, (x,y_1) \in \Pi_1, \dots, (x,y_k) \in \Pi_k  \}.
\]
The following simple lemma will be mainly applied for the case that the $\Pi_i$ are languages
and $\tau$ is a boolean function.

\begin{lemma}
	\label{lem:comb-streaming-algo}
	Let $R_1, \dots, R_k$ be randomized streaming algorithms and let $\Pi = \tau(\Pi_1, \dots, \Pi_k)$.
	Then there exists a randomized streaming algorithm $R$ such that for all $w \in \Sigma^*$:
	\begin{itemize}
		\item $\mathrm{space}(R,w) \le 2 \cdot \sum_{i=1}^k \mathrm{space}(R_i,w)$
		\item $\phi(R,w,\Pi,\eps) \le \sum_{i=1}^k \phi(R_i,w,\Pi_i,\eps_i)$
		for all $0 \le \eps_1, \dots, \eps_k \le 1$ with $\eps = \sum_{i=1}^k \eps_i \le 1$.
	\end{itemize}
\end{lemma}

\begin{proof}
	The algorithm $R$ simulates the $k$ algorithms $R_i$ in parallel with independent random bits.
	For this $2 \cdot \sum_{i=1}^k \mathrm{space}(R_i,w)$ bits are sufficient (the encodings of the states of the $R_i$
	have to separated, which explains the factor $2$).
	If $R_i$ outputs some value $y_i$ for $1 \le i \le k$,
	then $R$ outputs $\tau(y_1, \dots, y_k)$.
	Consider $0 \le \eps_1, \dots, \eps_k \le 1$ with $\eps = \sum_{i=1}^k \eps_i \le 1$.
	Let $w = a_1 \cdots a_m \in \Sigma^*$.
	If $\eps(R_i,a_1 \cdots a_t,\Pi_i) \le \eps_i$ for all $i \in [1,k]$
	then $\eps(R,a_1 \cdots a_t,\Pi) \le \eps$ by the union bound,
	and therefore $\FT(R,w,\Pi,\eps) \subseteq \bigcup_{i=1}^k \FT(R_i,w,\Pi_i,\eps_i)$.
	This implies
	\begin{eqnarray*}
		\phi(R,w,\Pi,\eps) &=& \frac{|\FT(R,w,\Pi,\eps)|}{m+1} \\
		&\le & \frac{|\bigcup_{i=1}^k \FT(R_i,w,\Pi_i,\eps_i)|}{m+1} \\
		&\le & \sum_{i=1}^k \frac{|\FT(R_i,w,\Pi_i,\eps_i)|}{m+1} \\
		&=& \sum_{i=1}^k \phi(R_i,w,\Pi_i,\eps_i),
	\end{eqnarray*}
	which proves the lemma.
\end{proof}

\subsection{Sliding window algorithms}

For a window length $n \ge 0$ and a stream $x \in \Sigma^*$
we define $\last_n(x)$ to be the suffix of $\square^n x$ of length $n$
where $\square \in \Sigma$ is a fixed alphabet symbol.
The word $\last_n(\varepsilon) = \square^n$ is also called the {\em initial window}.
Given an approximation problem $\Pi \subseteq \Sigma^* \times \Omega$ and a window length $n \ge 0$
we define the sliding window problem
\[
	\Pi_n = \{ (x,y) \in \Sigma^* \times \Omega \colon (\last_n(x),y) \in \Pi \}.
\]
Since we view a language $L \subseteq \Sigma^*$ as particular approximation problems,
the definition of $\Pi_n$ specializes to
\begin{equation} \label{def-L_n}
L_n = \{ x \in \Sigma^* \colon \last_n(x) \in L \} .
\end{equation}
A {\em randomized sliding window algorithm} ({\em randomized \SWA} for short) is a sequence $\R = (R_n)_{n \ge 0}$
of randomized streaming algorithms $R_n$ over the same alphabet $\Sigma$ and over the same set of output values $\Omega$.
If every $R_n$ is deterministic, we speak of a {\em deterministic \SWA}.
The {\em space complexity} of the randomized \SWA\  $\R = (R_n)_{n \ge 0}$
is the function
\[
	f(\R,n) = \sup \{ \mathrm{space}(R_n,u) : u \in \Sigma^* \}
\]
and its {\em expected space complexity} is the function
\[
	f_\varnothing(\R,n) = \sup \{ s_\varnothing(R_n,u) : u \in \Sigma^* \}.
\]
Clearly, if $R_n$ is finite, then one can always find a state encoding such that $f(\R,n) = \lfloor \log_2 |R_n| \rfloor$. 

\begin{definition}
Let $0 \le \eps \le 1$ be an error bound and $0 \le \phi \le 1$, let $\R = (R_n)_{n \ge 0}$ be a randomized \SWA,
and let $\Pi$ be an approximation problem.
\begin{itemize}
	\item We say that $\R$ is {\em $(\eps,\phi)$-correct for $\Pi$}
	if $\phi(R_n,w,\Pi_n,\eps) \le \phi$ for all $n \ge 0$ and $w \in \Sigma^{\ge n}$.
	The number $\eps$ is the {\em error probability} and $\phi$ is the {\em failure ratio} of $\R$.
	\item We say that $\R$ is {\em strictly $\eps$-correct for $\Pi$}
	if $\eps_*(R_n,w,\Pi_n) \le \eps$ for all $n \ge 0$ and $w \in \Sigma^{\ge n}$.
\end{itemize}
\end{definition}
Note that the definition of an $(\eps,\phi)$-correct randomized \SWA\ $\R$ for $\Pi$
also makes sense in the special case that $\R$ is deterministic and $\epsilon=0$.
A $(0,\phi)$-correct deterministic \SWA\ $\R = (R_n)_{n \ge 0}$ for $\Pi$ has the property that 
$R_n$ produces at most $\phi \cdot (m+1)$ many incorrect outputs when running on any
input word of length $m \geq n$.


Let us also emphasize that our model of sliding window algorithms is non-uniform in the 
sense that for every window length $n$ we have a separate algorithm. This makes lower
bounds stronger. On the other hand, in our upper bounds, the constructed sliding window
algorithms are uniform, in the sense that one has a single algorithm that is parameterized
by the window length. 

\subsection{Basic properties of sliding window algorithms}

The following lemma is a direct consequence of Lemma~\ref{lem:amplification}.
It shows that, if the error probability is strictly below $\frac{1}{2}$,
then reducing the error probability further increases the space complexity only by a constant factor.
This justifies an arbitrary choice of $\eps = \frac{1}{3}$.

\begin{lemma}[probability amplification]
	\label{lem:fs-ampli}
	Let $L \subseteq \Sigma^*$, $0 <\eps' < \eps <\frac{1}{2}$ and $0 \leq \phi \leq 1$.
	Given a randomized \SWA\ $\R$ which is $(\eps,\phi)$-correct for $L$,
	one can construct a randomized \SWA\ $\R'$ which is $(\eps',\phi)$-correct for $\Pi$ such that
	$f(\R',n) \le \ln(\frac{1}{\eps'}) \cdot \frac{1}{\mathrm{poly}(\eps)} \cdot f(\R,n)$. 
\end{lemma}

\begin{definition}
        Let  $\Pi$ be an approximation problem and $0 \leq \phi \leq 1$.
        \begin{itemize}
\item A {\em  randomized \SWA\ for $\Pi$} with failure ratio $\phi$
	is a randomized \SWA\ which is $(1/3,\phi)$-correct for $\Pi$.
\item If moreover $\phi=0$, then 
        we speak of a {\em randomized \SWA\ for $\Pi$}.  
\item A {\em  deterministic \SWA\ for $\Pi$} with failure ratio $\phi$
	is a deterministic \SWA\ which is $(0,\phi)$-correct for $\Pi$. 
\item If moreover $\phi=0$, then 
        we speak of a {\em deterministic \SWA\ for $\Pi$}.
\end{itemize}
\end{definition}

\begin{lemma}
	\label{lem:fs-expected-worst}
	Let $\R$ be a randomized \SWA\ which is $(\eps,\phi)$-correct for $\Pi$
	and let $\mu \ge 1$.
	Then there exists a randomized \SWA\ $\R'$ which is $(\eps + \frac{1}{\mu},\phi)$-correct for $\Pi$ such that
	$f(\R',n) \le \mu \cdot f_\varnothing(\R,n)$.
\end{lemma}

\begin{proof}
	Fix an $n \ge 0$ and let $s = f_\varnothing(\R,n)$ be the expected space complexity on window length $n$.
	If $s = \infty$ then the statement is trivial (we can take $R'_n = R_n$). So, let us assume that $s$ is finite.
	Let $Q_{\ge}$ be the set of states in $R_n$ with encoding length $\ge \mu \cdot s$. 
	If $Q_{\ge} = \emptyset$, then $f(\R,n) \le \mu \cdot f_\varnothing(\R,n)$ already holds.
	If $Q_{\ge}$ is nonempty, let $R'_n$ be the algorithm obtained from $R_n$ by identifying all states $q \in Q_{\ge}$
	into a single state $q_\bot$ encoded using an unused bit string of minimal length, which is at most $\mu \cdot s$.
	On an input stream $w \in \Sigma^*$,
	the probability that $q_\bot$ is reached is
	\[
		\Pr_{\pi \in \Runs(R'_n,w)}[\pi \text{ contains } q_\bot] = \Pr_{\pi \in \Runs(R_n,w)}[\mathrm{space}(R_n,\pi) \ge \mu \cdot s] \le \frac{s}{\mu \cdot s} = \frac{1}{\mu}
	\]
	by Markov's inequality. Note that if an $R'_n$-run on $w$ is not correct, then (i) it must contain $q_\bot$ or (ii)
	it must be a non-correct $R_n$-run on $w$. Hence, a union bound yields
	\[
		\eps(R'_n,w,\Pi) \leq \frac{1}{\mu} + \eps(R_n,w,\Pi) .
	\]
	By taking the supremum over all $w \in \Sigma^*$ on both sides of the inequality, we obtain the lemma for the $\eps$-error.
\end{proof}
Lemma~\ref{lem:fs-ampli} and \ref{lem:fs-expected-worst} justify
to focus on the worst-case space complexity $f(\R,n)$ in the rest of the paper.
Moreover, we only consider randomized SWAs $\R = (R_n)_{n \ge 0}$
where every $R_n$ has a finite state set $Q_n$. This is justified by the fact that
for every language $L$ and every $n$ the language $L_n$ from \eqref{def-L_n} is regular
and hence can be accepted by a DFA. The space-optimal deterministic \SWA\ for a language
$L$ therefore consists of the minimal DFA for $L_n$ for every $n \geq 0$. 
For a fixed error probability $\epsilon<1/2$ a space-optimal randomized
 \SWA\ for $L$ consists of a minimal probabilistic finite automaton for $L_n$ with error probability  $\epsilon$
for every $n \geq 0$. This probabilistic finite automaton has an isolated
cut-point \cite{Rabin63} (meaning that there is a probability gap). 
Rabin has shown in \cite{Rabin63} that a  probabilistic finite automaton with a finite cut-point
can be transformed into an equivalent DFA with an exponential blow-up.
Hence, we get:

\begin{lemma}
Let $\R$ be a randomized \SWA\ for the language $L$. Then, there exists
a deterministic \SWA\ $\D$ for $L$ such that $f(\D,n) \in \mathcal{O}(2^{f(\R,n)})$.
\end{lemma}

Throughout the paper we use the simple fact that
space complexity classes in the sliding window sense
are closed under Boolean combinations, which follows from Lemma~\ref{lem:comb-streaming-algo}:

\begin{lemma}
	Let $L$ be a Boolean combination of languages $L_1, \dots, L_k$.
	For each $i \in [1,k]$ let $\R_i$ be a randomized SWA for $L_i$ with failure ratio $\phi_i$.
	Then $L$ has a SWA $\R$ with failure ratio $\sum_{i=1}^k \phi_i$ and
	$f(\R,n) = \mathcal{O}(\sum_{i=1}^k f(\R_i,n))$.
\end{lemma}

\section{Main results} \label{sec-main}

\begin{figure}

\tikzstyle{class}=[minimum size = 20pt]

\begin{tikzpicture}
	\node [class] (1a) {$\Reg$};
	\node [class, below = 20pt of 1a] (1b) {$\langle \LI,\Len \rangle$};
	\node [class, below = 60pt of 1b] (1c) {$\langle \ST,\Len \rangle$};

	\node [class, right = 40pt of 1a] (2a) {$\Reg$};
	\node [class, below = 20pt of 2a] (2b) {$\langle \LI,\Len \rangle$};
	\node [class, below = 20pt of 2b] (2c) {$\langle \ST, \SF, \Len \rangle$};
	\node [class, below = 20pt of 2c] (2d) {$\langle \ST,\Len \rangle$};
	
	\node [class, right = 50pt of 2a] (3a) {$\Reg$};
	\node [class, below = 20pt of 3a] (3b) {$\langle \LI, \PF, \Len \rangle$};
	\node [class, below = 60pt of 3b] (3c) {$\langle \LB, \PF, \SF, \Len \rangle$};

	\node [class, right = 50pt of 3a] (4a) {$\Reg$};
	\node [class, below = 100pt of 4a] (4b) {$\langle \LI, \PF, \Len \rangle$};

	\draw (1a) -- (1b) -- (1c);
	\draw (2a) -- (2b) -- (2c) -- (2d);
	\draw (3a) -- (3b) -- (3c);
	\draw (4a) -- (4b);

	\node [class, left = 40pt of 1a] (0a) {$\mathcal{O}(n)$};
	\node [class, below = 20pt of 0a] (0b) {$\mathcal{O}(\log n)$};
	\node [class, below = 20pt of 0b] (0c) {$\mathcal{O}(\log \log n)$};
	\node [class, below = 20pt of 0c] (0d) {$\mathcal{O}(1)$};

	\node [class, above = 10pt of 0a] {complexity};
	\node [class, above = 10pt of 1a] {det. ($\phi = 0$)};
	\node [class, above = 10pt of 2a] {rand. ($\phi = 0$)};
	\node [class, above = 10pt of 3a] {det. ($\phi \to 0$)};
	\node [class, above = 10pt of 4a] {rand. ($\phi \to 0$)};

\end{tikzpicture}

\caption{\label{fig-survey} All language classes are defined in Section~\ref{sec-classes}.}
\end{figure}
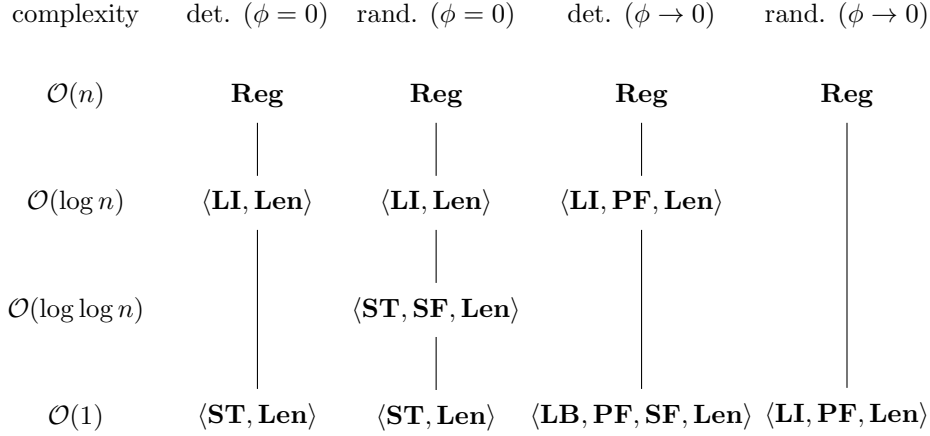

In this section we state the main results of this paper. 
We start with the randomized space complexity 
of regular languages in the sliding-window model with failure ratio
zero. This means that at every time instant the error probability
must be below $1/3$. The following
theorem gives a complete characterization. Points~\eqref{point-O(1)}
and \eqref{point-O(log)} have been shown already in \cite{GHL16}.

\begin{theorem}\label{thm:quatrochotomy}
	Let $L \subseteq \Sigma^*$ be a regular language.
	\begin{enumerate}
	\item \label{point-O(1)} If $L \in \langle \ST, \Len \rangle$,
	then $L$ has a deterministic \SWA\ $\R$ with $f(\R,n) = \mathcal{O}(1)$.
	\item \label{lower-loglog} If $L \notin \langle \ST, \Len \rangle$, then 
	$f(\R,n) \notin o(\log \log n)$ for every randomized \SWA\ $\R$ for $L$.
	\item \label{upper-loglog} If $L \in \langle \ST, \SF, \Len \rangle$,
	then $L$ has a randomized \SWA\ $\R$ with $f(\R,n) = \mathcal{O}(\log \log n)$.
	\item \label{lower-log} If $L \notin \langle \ST, \SF, \Len \rangle$, then $f(\R,n) \notin o(\log n)$
	for every randomized \SWA\ $\R$ for $L$.
	\item \label{point-O(log)} If $L \in \langle \LI, \Len \rangle$,
	then $L$ has a deterministic \SWA\ $\R$ with $f(\R,n) = \mathcal{O}(\log n)$.
	\item \label{lower-lin} If $L \notin \langle \LI, \Len \rangle$, then $f(\R,n) \notin o(n)$
	for every randomized \SWA\ $\R$ for $L$.
	\end{enumerate}
\end{theorem}
In the previous theorem, we only talk about randomized algorithms with failure ratio zero.
If we allow an arbitrarily small non-zero failure ratio we get the following space dichotomy.

\begin{theorem} \label{thm:dichotomy}
	Let $L \subseteq \Sigma^*$ be a regular language.
	\begin{enumerate}
	\item \label{upper-O(1)-failure} If $L \in \langle \LI, \PF, \Len \rangle$ and $0 < \phi \leq 1$,
	then $L$ has a randomized \SWA\ with $f(\R,n) = \mathcal{O}(1)$ and failure ratio $\phi$.
	\item \label{lower-O(n)-failure} If $L \notin \langle \LI, \PF, \Len \rangle$, then there exists a failure ratio $0 < \phi \leq 1$
	such that $f(\R,n) \notin o(n)$ for every randomized \SWA\ $\R$ for $L$ with failure ratio $\phi$.
	\end{enumerate}
\end{theorem}
Finally, for deterministic SWAs with an arbitrarily small non-zero failure ratio we get a space trichotomy:

\begin{theorem} \label{thm:trichotomy}
	Let $L \subseteq \Sigma^*$ be a regular language.
	\begin{enumerate}
	\item \label{upper-O(1)-failure-det} If $L \in \langle \LB, \PF, \SF, \Len \rangle$ and $0 < \phi \leq 1$,
	then $L$ has a deterministic \SWA\ with $f(\R,n) = \mathcal{O}(1)$ and failure ratio $\phi$.
	\item \label{lower-O(log n)-failure-det} If $L \notin \langle \LB, \PF, \SF, \Len \rangle$, then there exists a failure ratio $0 < \phi \leq 1$
	such that $f(\R,n) \notin o(\log n)$ for
	every deterministic \SWA\ $\R$ for $L$ with failure ratio $\phi$.
	\item \label{upper-O(log n)-failure-det} If $L \in \langle \LI, \PF, \Len \rangle$ and $0 < \phi \leq 1$,
	then $L$ has a deterministic \SWA\ with $f(\R,n) = \mathcal{O}(\log n)$ and failure ratio $\phi$.
	\item \label{lower-O(n)-failure-det} If $L \notin \langle \LI, \PF, \Len \rangle$, then there exists a failure ratio $0 < \phi \leq 1$
	such that $f(\R,n) \notin o(n)$ for every deterministic \SWA\ $\R$ for $L$ with failure ratio $\phi$.
	\end{enumerate}
\end{theorem}
Note that point \eqref{lower-O(n)-failure-det} from Theorem~\ref{thm:trichotomy} is an immediate
corollary of  Theorem~\ref{thm:dichotomy}\eqref{lower-O(n)-failure}. Also note that 
point \eqref{upper-O(log n)-failure-det} from Theorem~\ref{thm:trichotomy} follows from 
Theorem~\ref{thm:trichotomy}\eqref{upper-O(1)-failure-det} and 
Theorem~\ref{thm:quatrochotomy}\eqref{point-O(log)}.

\begin{figure}

\tikzstyle{class}=[minimum size = 20pt]

\begin{tikzpicture}

	\node [class] (2a) {$\Reg$};
	\node [class, below = 15pt of 2a] (2b) {$\langle \LI,\PF,\Len \rangle$};
	\node [class, below left = 15pt and -20pt of 2b] (2c) {$\langle \LI, \Len \rangle$};
	\node [class, below right = 15pt and -35pt of 2b] (2f) {$\langle \LB, \PF, \SF, \Len \rangle$};
	\node [class, below right = 15pt and -20pt of 2c] (2d) {$\langle \ST,\SF,\Len \rangle$};
	\node [class, below = 15pt of 2d] (2e) {$\langle \ST,\Len \rangle$};

	\draw (2a) -- (2b) -- (2c) -- (2d) -- (2e);
	\draw (2b) -- (2f) -- (2d);
	
	\draw (2a) ++ (1.5,0) node {$a\Sigma^*$};
	\draw (2b) ++ (3,0) node {$a\{a,b\}^*c \cup \{a,b\}^*$};
	\draw (2c) ++ (-1.5,0) node {$\Sigma^*a\Sigma^*$};
	\draw (2d) ++ (2,0) node {$ab^*$};
	\draw (2e) ++ (2,0) node {$\Sigma^*a$};
	\draw (2f) ++ (2,0) node {$a^*b$};

\end{tikzpicture}

\caption{\label{class-containment} Examples for all occurring language classes where $\Sigma = \{a,b,c\}$.}
\end{figure}
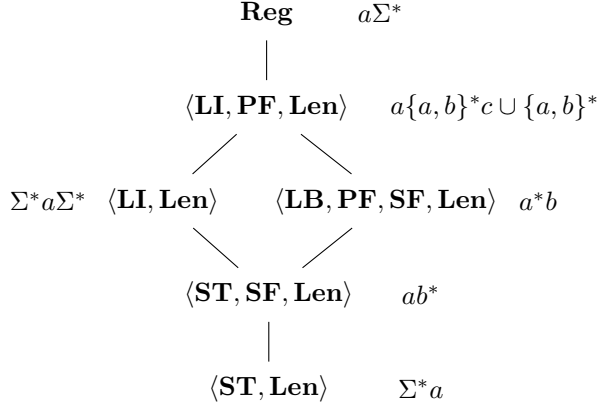

Figure~\ref{fig-survey} gives an overview on our main results.
Figure~\ref{class-containment} shows an inclusion diagram for the language classes in 
Figure~\ref{fig-survey}. One can show that the example languages in Figure~\ref{class-containment} 
witness the strictness of the inclusions.

\section{Upper bounds}

\subsection{Deterministic algorithms with arbitrarily small failure ratio} 

In this section, we prove the remaining upper bound \eqref{upper-O(1)-failure-det}
from Theorem~\ref{thm:trichotomy}. For this we consider the four base cases:
$L \in \LB$, $L \in \PF$,  $L \in \SF$, and $L \in \Len$. The case $L \in \Len$ is obvious:

\begin{lemma}
	\label{lem:length-lang}
	Every length language $L \subseteq \Sigma^*$ has a deterministic \SWA\ $\R$
	with $f(\R,n) = \mathcal{O}(1)$.
\end{lemma}
Next, we consider regular prefix-free languages.
Recall the definition of the word $w[i,j]$ for a word $w$ and a (possibly empty) interval $[i,j]$; see Section~\ref{prel}.
Let $S$ be a set of subintervals of $[1,m]$.
We call $S$ {\em overlapping} if $\bigcap_{I \in S} I \neq \emptyset$
and $S$ {\em increasing} if for all $[i,j], [k,\ell] \in S$, $j=\ell$ implies we have $i=k$;
see Figure~\ref{fig-overlapping} for an illustration.

\begin{figure}[h]
	\begin{tikzpicture}
		\draw[|-|] (0,1.2) -- (2,1.2);
		\draw[|-|] (1,0.9) -- (2.5,0.9);
		\draw[|-|] (-0.5,0.6) -- (3,0.6);
		\draw[-, thick] (-0.5,0.6) -- (3,0.6);
		\draw[|-|] (0.5,0.3) -- (3.5,0.3);
		\draw[|-|] (0,0) -- (4,0);
		\draw[-, thick] (1.5,0) -- (4,0);

		\node[fill, circle, inner sep = 1pt] (a) at (1.5,0.6) {};
		\node[fill, circle, inner sep = 1pt] (b) at (1.5,0) {};
		\draw[densely dotted] (a) -- (b);
	\end{tikzpicture}
	\caption{\label{fig-overlapping} An overlapping increasing set of words.}
\end{figure}
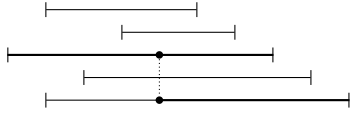

\begin{lemma}
	\label{lem:overlap}
	Let $L \in \PF$, $A = (Q,\Sigma,q_0,\delta,F)$ be a DFA for $L$, $w = a_1 \cdots a_m \in \Sigma^*$
	and $S$ be a set of subintervals of $[1,m]$, 
	which is increasing and overlapping.
	If $\{ w[i,j] \colon [i,j] \in S \} \subseteq L$ then $|S| \le |Q|$.
\end{lemma}

\begin{proof}
	Assume that $|S| \ge |Q|+1$ and let $\ell \in \bigcap_{I \in S} I$.
	Thus, for every $[i,j] \in S$, $w[i,\ell]$ is a prefix of $w[i,j]$.
	By the pigeonhole principle there are two distinct intervalls
	$[\ell_1, r_1], [\ell_2, r_2] \in S$ such that
	$\delta(q_0, w[\ell_1,\ell]) = \delta(q_0, w[\ell_2,\ell])$.
	This implies $\delta(q_0, w[\ell_1,r_1]) = \delta(q_0, w[\ell_2,r_1])$.
	Since $w[\ell_1,r_1] \in L$, we also have $w[\ell_2,r_1] \in L$.
	Since $S$ is increasing, we have $r_1 \neq r_2$; assume w.l.o.g.~that $r_1 < r_2$.
	But then $w[\ell_2,r_1]$ is a proper prefix of $w[\ell_2,r_2]$.
	Since both $w[\ell_2,r_1]$ and $w[\ell_2,r_2]$  belong to $L$,
	this contradicts the prefix-freeness of $L$.
\end{proof}

\begin{theorem}
	\label{thm:prefix-free}
	Let $0 < \phi \le 1$ and let $L \in \PF$.
	Then $L$ has a deterministic \SWA\ with $f(\R,n) = \mathcal{O}(1)$ and failure ratio $\phi$.
\end{theorem}

\begin{proof}
	Let $\R = (R_n)_{n \ge 0}$ where $R_n$ is the one-state automaton which always rejects.
	Furthermore let $A = (Q,\Sigma,q_0,\delta,F)$ be a DFA for $L$.
	Let $n \ge 0$ be a window size and $x = a_1 \cdots a_m$ be an input stream where $m \ge n$.
	We have
	\[
		\phi(R_n,x,L_n,0) = \frac{1}{m+1} \cdot |\{ t \in [0,m] : \last_n(a_1 \cdots a_t) \in L \}| .
	\]
	Consider an interval $I = [i,j] \subseteq [0,m]$ of size $j-i+1 \leq n$.
	By Lemma~\ref{lem:overlap} there are at most $|Q|$ many indices $t \in I$ such that $\last_n(a_1 \cdots a_t) \in L$.
	Since we can partition $[0,m]$ into  $\lceil \frac{m+1}{n} \rceil$ many intervals of length at most $n$,
	we have
	\begin{eqnarray*}
		\frac{1}{m+1} \cdot |\{ i \in [0,m] : \last_n(a_1 \cdots a_i) \in L \}|
		&\le& \frac{|Q|}{m+1} \cdot \left\lceil \frac{m+1}{n} \right\rceil \\
		&\le & \frac{|Q|}{m+1} \cdot \left(\frac{m+1}{n} + 1\right) \\
		&=& \frac{|Q|}{n} + \frac{|Q|}{m+1} \le \frac{2|Q|}{n}.
	\end{eqnarray*}
	Since this number converges to $0$ for increasing $n$, there exists a window size $n_0$ such that for all $n \ge n_0$
	and all $x \in \Sigma^{\ge n}$ we have $\phi(R_n,L_n,x,0) \le \phi$.
	Finally, for all $n < n_0$ we replace the algorithm $R_n$ by a trivial deterministic DFA for $L_n$,
	to obtain the failure ratio for all window lengths.
	This concludes the proof. 
\end{proof}
The arguments from the proof of Theorem~\ref{thm:prefix-free}
can be used for suffix-free regular languages as well:

\begin{theorem}
	\label{thm:suffix-free}
	Let $0 < \phi \le 1$ and let $L \in \SF$.
	Then $L$ has a deterministic \SWA\ with $f(\R,n) = \mathcal{O}(1)$ and failure ratio $\phi$.
\end{theorem}

\begin{proof}
	Let $A = (Q,\Sigma,q_0,\delta,F)$ be a DFA for $L^\rev$.
	Again we consider a window size $n \ge 0$, an input stream $x = a_1 \cdots a_m$ of length at least $n$,
	and an interval $I  \subseteq [0,m]$ of size at most $n$.
	The same argument as above shows that there are at most $|Q|$ many indices $i \in I$
	such that $\last_n(a_1 \cdots a_i)^\rev \in L^\rev$, or equivalently $\last_n(a_1 \cdots a_i) \in L$.
	We can now conclude as in the proof of Theorem~\ref{thm:prefix-free}.
\end{proof}
Finally, we consider the remaining case of a language from $\LB$:
\begin{theorem}
	Let $0 < \phi \le 1$ and $L \in \LB$.
	Then $L$ has a deterministic \SWA\ with $f(\R,n) = \mathcal{O}(1)$ and failure ratio $\phi$.
\end{theorem}

\begin{proof}
	Let $A = (Q,\Sigma,q_0,\delta,F)$ be a DFA for $L$ and assume that $L = \Sigma^* K$ where $K \in \PF$. 
	Basically we will use the DFA $A$ itself as a sliding window algorithm.
	
	Let $n \ge 0$ be a window size and $x = a_1 \cdots a_m$ be an input stream where $m \ge n$.
 	Define $R_n = (Q,\Sigma,\delta(q_0,\Box^n),\delta, F)$.
	Setting $y := \Box^n x$, we have:
	\begin{equation}
		\label{eq:Rn-accept}
		R_n \text{ accepts } x[1,t] \iff y[1,n+t]  \in L \iff \exists i \in [1,n+t+1] :  y[i,n+t] \in K
	\end{equation}
	Clearly, if a window $\last_n(x[1,t])$ belongs to $L$
	then also $y[1,n+t]$ belongs to $L$
	because $\last_n(x[1,t])= y[t+1,n+t]$ is a suffix of $y[1,n+t]$ and $L$ is a left-ideal.
	Hence $R_n$ accepts $x[1,t]$. 
	That means that $R_n$ only makes false positive errors, i.e.,
	\[
		R_n \text{ fails at } t \in [0,m] \iff \last_n(x[1,t]) \notin L \text{ and } R_n \text{ accepts } x[1,t].
	\]
	If $R_n$ fails at time instant $t$ then there exists a number $i \in [1,n+t+1]$ such that $y[i,n+t] \in K$ by \eqref{eq:Rn-accept}.
	Furthermore we know that $i \leq t$ because, otherwise $y[i,n+t]$ would be a suffix
	of $\last_n(x[1,t]) = y[t+1,n+t]$ which does not belong to $L = \Sigma^* \cdot K$ by assumption.
	
	Now consider an interval $J = [k,\ell] \subseteq [0,m]$ of size at most $n$
	and let $\FT_J$ be the set of time instants $t \in J$ at which $R_n$ fails on $x$.
	Let $S$ be a set of intervals which contains for each $t \in\FT_J$
	exactly one interval $[i,n+t]$ such that $y[i,n+t] \in K$.
	Clearly, $S$ is increasing.
	Since each interval in $S$ has size at least $n+1$, the point $k$ is contained in each interval in $S$,
	i.e., $S$ is overlapping.
	By Lemma~\ref{lem:overlap} we know $|S| \le |Q|$ and hence $|\FT_J| \le |Q|$.
	The rest of proof follows the proof of Theorem~\ref{thm:prefix-free}.
\end{proof}

\subsection{The Bernoulli algorithm} \label{sec-bernoulli}

In this section, we introduce a randomized \SWA\ that will be used for the proof of
the upper bounds \eqref{upper-loglog} from Theorem~\ref{thm:quatrochotomy}
and \eqref{upper-O(1)-failure} from Theorem~\ref{thm:dichotomy}.

Consider an regular language $L \subseteq \Sigma^*$ and let $A = (Q,\Sigma,q_0,\delta,F)$
be a DFA for $L^\rev$.
For a stream $w \in \Sigma^*$ define the function $\ell_w \colon Q \to \mathbb{N} \cup \{ \infty \}$ by
\begin{equation}
 \ell_w(q) = \inf \{  k \in \mathbb{N} : \delta(q,\last_k(w)^\rev) \in F \}, \label{def-l_q(w)}
\end{equation}
where we set $\inf(\emptyset) = \infty$.
One can define a deterministic \SWA\ which stores the function $\ell_w$ on input stream $w$.
If a symbol $a \in \Sigma$ is read, we can determine
\[
	\ell_{wa}(q) = \begin{cases}
	0, & \text{if } q \in F, \\
	1 + \ell_w(\delta(q,a)), & \text{otherwise,}
	\end{cases}
\]
where $1 + \infty = \infty$. 

We will use the values $\ell_w(q)$ in case $L$ is a left ideal or suffix-free. In these cases, the value 
$\ell_w(q_0)$ can be used to decide whether $\last_n(w) \in L$:
\begin{itemize}
\item If $L$ is a left ideal, then $\last_n(w) \in L$ if and only if $\ell_{w}(q_0) \le n$.
\item If $L$ is suffix-free, then $\last_n(w) \in L$ if and only if $\ell_{w}(q_0) = n$.
\end{itemize}
Using a Bernoulli random variable, we define a randomized approximation
of the above deterministic \SWA.
Let $\beta \colon \mathbb{N} \to \mathbb{R}$ be a function such that
for some $n_0$, $0 \leq \beta(n) \leq 1$ for all $n \geq n_0$. Later, the
function $\beta$ will be instantiated by concrete functions.
In the following, we always set $x^{\infty} = 0$ for $0 \leq x < 1$.
We define the following constant-space randomized \SWA\ $\mathcal{B} = (B_n)_{n \ge 0}$
(which depends on the language $L$, the DFA $A$ and the function $\beta$), which we will call
the Bernoulli algorithm.
If $n < n_0$ let $B_n$ be the trivial deterministic streaming algorithm for $L_n$.
For $n \ge n_0$ the algorithm $B_n$ stores a Boolean flag for each state
in form of a function $b \colon Q \to \{0,1\}$. All flags $b(q)$ for $q \in F$ are fixed to $1$ forever.
For all other states $q \in Q \setminus F$ we define the initial value of the flag $b(q)$
as follows, where $\ell := \ell_{\varepsilon}(q) =
\inf \{  k \in \mathbb{N} : \delta(q,\Box^k) \in F \}
\in \mathbb{N} \cup \{\infty\}$:
\begin{equation} \label{eq-init}
b(q) := \begin{cases}
0  \text{ with probability } 1 - \left(1 - \beta(n)\right)^\ell \\
1 \text{ with probability } \left(1 - \beta(n)\right)^\ell .
\end{cases}
\end{equation}
For all states $q \in Q \setminus F$ we do the following upon
arrival of  a symbol $a \in \Sigma$:
\begin{equation} \label{eq-next}
b(q) := \begin{cases}
0 & \text{with probability } \beta(n)  \\
b(\delta(q,a)) & \text{with probability } 1-\beta(n) 
\end{cases}
\end{equation}
The algorithm accepts if and only if $b(q_0) = 1$. 

\begin{lemma} \label{lemma-error-prob1}
Let $n \ge n_0$, let $w \in \Sigma^*$ be an input stream and let $\ell = \ell_w(q_0)$. 
We have
\[
\Pr[B_n \text{ accepts } w] = (1 - \beta(n))^\ell .
\]
\end{lemma}

\begin{proof}
Let us define the $\{0,1\}$-valued random variable $B_{q,w}$ for 
$q \in Q$ and $w \in \Sigma^*$ as the value of the flag $b(q)$  after reading
the input stream $w$. 
We show the following statement for all $q \in Q$ and $w \in \Sigma^*$ with $k = \ell_w(q)$:
\begin{equation} \label{eq-Prob-B=1}
\Pr[B_{q,w} = 1]  = \left(1 - \beta(n)\right)^k
\end{equation}
This implies the statement of the lemma.

We prove \eqref{eq-Prob-B=1} by induction on $|w|$. Let $\ell = \ell_{w}(q)$.
First assume that $w=\varepsilon$. 
By the initialization of the flags $b(q)$ in \eqref{eq-init} we have
\[
\Pr[B_{q,\varepsilon} = 1]  = \left(1 - \beta(n)\right)^\ell
\]
Note that this is also true for $q \in F$ since $\ell_w(q) = 0$.

Let us now assume that $w = w'a$. Let $q' = \delta(q,a)$ and $\ell' = \ell_{w'}(q')$.
If $q \in F$ then 
\[
\Pr[B_{q,w} = 1]  = 1 = \left(1 - \beta(n)\right)^\ell
\]
since $\ell_w(q) = 0$. Now assume that $q \not\in F$ and thus $\ell = \ell'+1$.
We get by induction and \eqref{eq-next}
\begin{eqnarray*}
\Pr[B_{q,w} = 1] 
&=&  \Pr[B_{p,w'} = 1] \cdot \left( 1-\beta(n)  \right) \\
&=& \left(1 - \beta(n)\right)^{\ell'} \cdot \left( 1-\beta(n)  \right) \\
&=& \left(1 - \beta(n)\right)^{\ell}.
\end{eqnarray*}
This concludes the proof of the lemma.
\end{proof}


\subsection{Zero-failure randomized algorithms} 

In this section we prove the remaining upper bound
\eqref{upper-loglog} from Theorem~\ref{thm:quatrochotomy}
(recall that \eqref{point-O(1)}
and \eqref{point-O(log)} have been shown in \cite{GHL16}):

\begin{theorem}
If $L \in \langle \ST, \SF, \Len \rangle$,
then $L$ has a randomized \SWA\ $\R$ with $f(\R,n) = \mathcal{O}(\log\log n)$.
\end{theorem}

\begin{proof}
Since languages in $\ST \cup \Len$ have constant space deterministic \SWAs\
(Theorem~\ref{thm:quatrochotomy}, point \eqref{point-O(1)}),
it suffices to prove the statement for $L \in \SF$.
Let $A = (Q,\Sigma,q_0,\delta,F)$ be a DFA for $L^\rev \in \PF$.
Since the case $L = \emptyset$ is trivial,
we can assume that $A$ contains at least one final state which is reachable from $q_0$.
Furthermore, since $L^\rev$ is prefix-free, any run in $A$ from $q_0$ contains at most one final state.
Therefore, we can assume that $F$ contains exactly one final state $q_F$, and all outgoing transitions
from $q_F$ lead to a sink state.

Recall the function $\ell_w \colon Q \to \mathbb{N} \cup \{\infty\}$ defined in  \eqref{def-l_q(w)}.
Notice that $\last_n(w) \in L$ if and only if $\ell_{w}(q_0) = n$ for all $w \in \Sigma^*$.
Our randomized \SWA\ $\R = (R_n)_{n \ge 0}$ consists of two parts:
a constant-space Bernoulli algorithm $\mathcal{B} = (B_n)_{n \ge 0}$ (see Section~\ref{sec-bernoulli}) 
which rejects with high probability whenever $\ell_{w}(q_0) \geq 2n$
and a modulo-counting algorithm $\mathcal{M} = (M_n)_{n \ge 0}$ which maintains $\ell_w$ modulo a random prime number
with $\mathcal{O}(\log \log n)$ bits.

The threshold algorithm is our Bernoulli algorithm $\mathcal{B} = (B_n)_{n \ge 0}$ 
from Section~\ref{sec-bernoulli} for the function $\beta(n) = 1/(2n)$ and $n_0 = 1$.
From Lemma~\ref{lemma-error-prob1} we get the following for all
$w \in \Sigma^*$ with $\ell = \ell_w(q_0)$. 
\begin{equation} \label{eq-bernoulli-accept}
\Pr[\text{$B_n$ accepts $w$}] = \left(1 - \frac{1}{2n}\right)^\ell .
\end{equation}
Let us now explain the modulo-counting algorithm $M_n$.
Let $p_i$ be the $i$-th prime number and let $s(m)$ be the product
of all prime numbers $\le m$. It is known that $\ln(s(m)) > m \cdot (1 - 1/\ln m)$ for $m \geq 41$
\cite[3.16]{RoScho62} and $p_i < i \cdot  (\ln i + \ln\ln i)$ for $i \geq 6$ \cite[3.13]{RoScho62}.
Let $k$ be the first natural number such that $\prod_{i=1}^k p_i \ge n$.
By the above bounds we get $k \in \mathcal{O}(\log n)$ and $p_{3k} \in \mathcal{O}(\log n \cdot \log \log n)$.
The algorithm $M_n$ initially picks a random prime $p \in \{p_1, \dots, p_{3k}\}$,
which is stored throughout the run using $\mathcal{O}(\log \log n)$ bits.
Then, after reading  $w \in \Sigma^*$, $M_n$ stores for every $q \in Q$ a bit telling whether $\ell_w(q) < \infty$ and, if the latter holds,
the values $\ell_w(q) \bmod p$ using $\mathcal{O}(|Q| \cdot \log \log n)$ bits. The algorithm accepts if and only if $\ell_w(q_0) \equiv n \bmod p$.

The combined algorithm $R_n$ accepts if and only if both $B_n$ and $T_n$ accept. 
Let us bound the error probability on an input stream $w \in \Sigma^*$ with $\ell = \ell_w(q_0)$.

\medskip
\noindent
{\em Case 1.} $\ell = n$, i.e., $\last_n(w) \in L$. Then  
$M_n$ accepts $w$ with probability 1. Moreover, by \eqref{eq-bernoulli-accept}, $B_n$ accepts 
$w$ with probability 
\[
\left(1 - \frac{1}{2n}\right)^n \geq 0.6 
\]
for $n \geq 12$
(note that $(1 - \frac{1}{2n})^n$ converges to $1/\sqrt{e} \approx 0.60653$ from below).
Hence, $R_n$ accepts with probability at least $0.6$.

\medskip
\noindent
{\em Case 2.} $\ell \geq 2n$ and hence $\last_n(w) \notin L$.
Then $B_n$ rejects with probability 
\[
1 - \left(1 - \frac{1}{2n}\right)^{\ell} \geq 1 - \left(1 - \frac{1}{2n}\right)^{2n} \geq 1 - 1/e \geq 0.6 .
\]
Here, we use the well-known inequality $(1-\frac{1}{y})^y \le e^{-1}$ for all $y \ge 1$.
Hence, $R_n$ also rejects with probability at least $0.6$.

\medskip
\noindent
{\em Case 3.} $\ell < 2n$ and $\ell \neq n$, and thus $\last_n(w) \notin L$.
Since $\ell - n \in [-n,n]$ and any product of at least $k+1$ pairwise distinct primes
exceeds $n$, the number $\ell - n \neq 0$ has at most $k$ prime factors.
Therefore, $M_n$ (and thus $R_n$) rejects with probability at least $2/3$.
\end{proof}

\subsection{Randomized algorithms with arbitrarily small non-zero failure ratio} 

In this section we prove \eqref{upper-O(1)-failure} from Theorem~\ref{thm:dichotomy}.
Since languages in $\PF \cup \Len$ have constant space deterministic \SWAs\ with 
arbitrarily small non-zero failure ratio (Lemma~\ref{lem:length-lang} and Theorem~\ref{thm:prefix-free}), 
it suffices to present a constant-space randomized \SWA\ with an arbitrarily small failure ratio $\phi>0$ for every regular left ideal.

For this subsection let $L \subseteq \Sigma^*$ be a regular left ideal.
Let $A = (Q,\Sigma,q_0,\delta,F)$ be the minimal DFA for $L^\rev$.
Since the case $L = \emptyset$ is trivial, we can assume that $L \neq \emptyset$.
It is easy to see that $F$ contains a single state $q_F$ from which all outgoing transitions lead back to $q_F$.
Recall the function $\ell_w \colon Q \to \mathbb{N} \cup \{\infty\}$ defined in \eqref{def-l_q(w)}.
Since $L$ is a left ideal, we have: $\last_n(w) \in L$ if and only if $\ell_w(q_0) \leq n$:
The following lemma says that the portion of prefixes of an input stream, where
$\ell_{w}(q_0)$ is close to $n$, is small:

\begin{lemma} \label{lemma-l_v(q_0)}
Let $0 < \xi < 1$. Let $n \ge 0$ be a window size and  $w \in \Sigma^{\geq n}$ be an input stream.
Then the number of prefixes $v$ of $w$
such that $\lceil \xi n\rceil \le \ell_{v}(q_0) \le n$ is at most
\[
	\frac{(1-\xi+\frac{1}{n}) \cdot |Q|}{\xi} \cdot (|w|+1+\xi n).
\]
\end{lemma}

\begin{proof}
Let us say that a prefix $v$ of $w$ is a {\em hit}, if $\lceil \xi n\rceil \le \ell_{v}(q_0) \le n$.
Consider an interval $I = [i, i']$ with $0 \leq i \leq i' \leq  |w|$
and $i'-i \leq \lceil \xi n \rceil$.
With each position $j \in I$ we associate the prefix $w[1,j]$. 

\medskip
\noindent
{\em Claim:} The set $V := \{ w[1,j] \colon j \in I \}$ contains at most $|Q|  \cdot (n - \lceil \xi n \rceil + 1)$
hits. 

\medskip
\noindent
In order to get a contradiction, 
assume that there are more than $|Q| \cdot (n - \lceil \xi n \rceil + 1)$ 
many hits in $V$. Since $\lceil \xi n\rceil \le \ell_{v}(q_0) \le n$ for every hit $v$,
and the interval $[\lceil \xi n\rceil, n] $ contains $n - \lceil \xi n \rceil + 1$ many different values,
there is a subset $U \subseteq V$ and some $\ell \in [\lceil \xi n\rceil, n] $ such that
(i) $|U| > |Q|$ and (ii) $\ell_{v}(q_0) = \ell$ for all $v \in U$.
Let $U = \{ w[1,j_1], w[1,j_2], \ldots, w[1,j_k] \}$, where $k > |Q|$.
Consider the words $u_1 = \last_\ell(w[1,j_1]), 
u_2 = \last_\ell(w[1, j_2]),  \ldots, u_k = \last_\ell(w[1, j_k])$. Since $j_k - j_1 \leq \lceil \xi n \rceil$
and $\ell \geq \lceil \xi n \rceil$, we have $\ell - j_k +j_1 \geq 0$.
Hence, we can consider the words $v_1 = \last_{\ell-j_k+j_1}(w[1, j_1]), 
v_2 = \last_{\ell-j_k+j_2}(w[1, j_2]),  \ldots, v_k = \last_\ell(w[1, j_k])$. 
Clearly, $v_j$ is a suffix of $u_j$. Moreover, the words $v_j$
all start in the same position of $\Box^n w$, i.e., every $v_j$ is a prefix of 
$v_{j'}$ for $j  \leq j'$. Consider now the state $q_j = \xi(q_0,v_j^\rev)$
for $1 \leq j \leq k$. Since $k > |Q|$ there exist $j < j'$ such that
$q_j = q_{j'}$. But this would imply that $\ell_{v_j}(q_0) <  \ell_{v_{j'}}(q_0)$,
which contradicts $\ell_{v_j}(q_0) = \ell =  \ell_{v_{j'}}(q_0)$. This concludes
the proof of the above claim.

Now we can finish the proof of the lemma: We divide the interval $[0,|w|]$ into intervals of size
$\lceil \xi n \rceil+1$ plus one last interval of possibly shorter length.
This yields 
\[
\left\lceil \frac{|w|+1}{\lceil \xi n \rceil+1} \right\rceil
\]
many intervals. In each of these intervals we find at most $|Q|  \cdot (n - \lceil \xi n \rceil + 1)$ many hits
by the above claim. Hence, the total number of hits is bounded by
\begin{eqnarray*}
\left\lceil \frac{|w|+1}{\lceil \xi n \rceil+1} \right\rceil \cdot |Q|  \cdot (n - \lceil \xi n \rceil + 1) 
 & \leq &
\left( \frac{|w|+1}{\xi n} +1\right) \cdot |Q|  \cdot (n - \xi n + 1) \\
& = & (|w|+1+\xi n) \cdot |Q| \cdot \frac{n - \xi n + 1}{\xi n} \\
& = & (|w|+1+\xi n) \cdot |Q| \cdot \frac{1 - \xi  + \frac{1}{n}}{\xi} .
\end{eqnarray*}
This concludes the proof of the lemma.
\end{proof}
We now define the function $\beta \colon \mathbb{N} \to \mathbb{R}$ by 
\[
\beta_\eps(n) = \frac{\ln(1/\eps)}{n} .
\]
Hence, for all $n \geq \ln(1/\eps)$ we have $0 < \beta_\eps(n) \leq 1$. 
We then consider the Bernoulli algorithm $\mathcal{B^{\eps}} = (B^{\eps}_n)_{n \geq 0}$ for $\beta_\eps$
defined in Section~\ref{sec-bernoulli}. Lemma~\ref{lemma-error-prob1}
gives the following guarantees for the error probability, where
$n \ge \ln(1/\eps)$, $w \in \Sigma^*$ and $k = \ell_w(q_0)$. 
\begin{equation} \label{eq-guarantees}
\eps(B_n,w,L_n) = \begin{cases}
1 - \left(1 - \frac{\ln(1/\eps)}{n}\right)^k & \text{ if } k \le n, \\
\left(1 - \frac{\ln(1/\eps)}{n}\right)^k & \text{ if } k > n.
\end{cases}
\end{equation}

\begin{lemma} \label{lemma-error-prob2}
For every $0 <  \xi <1$ there exists $0 < \eps < \frac{1}{2}$
and $n_0 \ge 1$ such that for all $n \geq n_0$ the following holds:
If $w \in \Sigma^*$ and $\ell_w(q_0) \not\in [\lceil \xi n\rceil, n]$, then $\eps(B^{\eps}_n,w,L_n) \le \eps$.
\end{lemma}

\begin{proof}
For a given $0 < \xi <1$ we consider the function
\begin{align*}
	g(x) =  x^\xi + x - 1
\end{align*}
for $0 < x < \frac{1}{2}$. Since $g$ is continuous, $\lim_{x \to \frac{1}{2}^-} g(x) > 0$
and $\lim_{x \to 0^+} g(x) = -1$
there exists $0 < \eps < \frac{1}{2}$
such that $g(\eps) > 0$, or equivalently $1-\eps < \eps^\xi$.
Choose such an $\eps$. 

Next, we determine the  number $n_0$ from the lemma.
Since
\begin{eqnarray*}
\lim_{n \to \infty} \left(1-\frac{\ln(1/\eps)}{n}\right)^{\xi n} &=& 
	\lim_{n \to \infty} \left(1-\frac{1}{n/\ln(1/\eps)}\right)^{\xi \cdot \ln(1/\eps) \cdot n/\ln(1/\eps)} \\
	&=& e^{-\xi \cdot \ln(1/\eps)} = \eps^\xi > 1-\eps
\end{eqnarray*}
for all $0 < \eps <\frac{1}{2}$,
there exists a natural number $n_1$ such that for all $n \ge n_1$ we have
\begin{equation} \label{eq-1-eps}
1 - \eps \le \left(1-\frac{\ln(1/\eps)}{n}\right)^{\xi n}.
\end{equation}
Let $n_0 = \max\{n_1, \lceil  \ln(1/\eps) \rceil \}$.
 
We now show that the error probabilities from \eqref{eq-guarantees} is bounded by $\eps$
whenever $n \geq n_0$ and $k = \ell_w(q_0) \not\in [\lceil \xi n\rceil, n]$.

\medskip
\noindent
{\em Case 1.} $k > n$: We use the well-known inequality $(1-\frac{1}{y})^y \le e^{-1}$ for all $y \ge 1$.
Equation~\eqref{eq-guarantees} together with $n/\ln(1/\eps) \geq 1$
yields the following bound on the error probability:
\[
\bigg(1 - \frac{\ln(1/\eps)}{n}\bigg)^k \leq \bigg(1 - \frac{\ln(1/\eps)}{n}\bigg)^{n} 
=  \bigg(1 - \frac{1}{n/\ln(1/\eps)}\bigg)^{n} 
\leq e^{-\ln(1/\eps)} = \eps.
\]
{\em Case 2.} $k < \lceil \xi n\rceil$, i.e., $k \leq \xi n$:
Again by \eqref{eq-guarantees} the error probability is bounded by
\[
	1 - \left(1 - \frac{\ln(1/\eps)}{n} \right)^k \le 1 - \left(1 - \frac{\ln(1/\eps)}{n} \right)^{\xi n} \le \eps ,
\]
where the last inequality follows from \eqref{eq-1-eps}.
\end{proof}

\begin{theorem}
Let $L$ be a regular left ideal and $0 < \phi < 1$.
Then $L$ has a randomized \SWA\ $\mathcal{R}$ with $f(\mathcal{R},n) = \mathcal{O}(1)$ and failure ratio $\phi$.
\end{theorem}

\begin{proof}
Let us fix a failure ratio $0 < \phi < 1$ and
let $0 <\xi<1$, which will be defined later (depending on $\phi$).
Let $\eps$ and $n_0$ be the numbers from Lemma~\ref{lemma-error-prob2}.
Let $\mathcal{B}^{\eps} = (B^{\eps}_n)_{n \ge 0}$ be the randomized \SWA\ described above.
Let $n \ge n_0$ be a window size and $w \in \Sigma^{\ge n}$ be an input stream.
Consider the set $P(w)$ of all prefixes $v$ of $w$ such that
$\ell_v(q_0) \in [\lceil \xi n\rceil, n] $. By Lemma~\ref{lemma-error-prob2} 
the algorithm $B_n^{\eps}$ errs on each prefix $v \notin P(w)$ with probability at most $\eps$,
i.e.,
\[
	\phi(B_n^{\eps},w,L_n,\eps) \le \frac{|P(w)|}{|w|+1}.
\]
Moreover, by Lemma~\ref{lemma-l_v(q_0)} we have
\[
|P(w)| \leq \frac{(1-\xi+\frac{1}{n}) \cdot |Q|}{\xi} \cdot (|w|+1+\xi n).
\]
We therefore get
\begin{eqnarray} 
	\phi(B_n^{\eps},w,L_n,\eps) &\le&  \frac{(1-\xi+\frac{1}{n}) \cdot |Q|}{\xi} \cdot \left(1+\frac{\xi n}{|w|+1}\right) \nonumber \\
	&\le& \left(1-\xi+\frac{1}{n}\right) \cdot |Q| \cdot \left(1+\frac{1}{\xi}\right). \label{prob-P(w)}
\end{eqnarray}
Note that if $\xi$ converges to $1$, then the probability \eqref{prob-P(w)}
tends towards $2|Q|/n$.
Hence we can choose numbers $n_1 \ge n_0$ and $0 < \xi < 1$ such that
for all $n \ge n_1$ the probability \eqref{prob-P(w)} 
is smaller than our fixed failure ratio $\phi$.

Finally for window sizes $n < n_1$ we can use the optimal deterministic sliding-window algorithms
for $L$ and window size $n$. The space complexity of the resulting algorithm is a constant that
depends only on $\phi$.
\end{proof}

\section{Lower bounds}

In this section, we prove the lower bounds from our three main results Theorem~\ref{thm:quatrochotomy}--\ref{thm:trichotomy}.
In all cases with one exception, we apply the same proof strategy. We first show that if a regular language does not belong to the language class
under consideration then there exist certain witness words. These witness words can then be used to apply known lower
bounds from communication complexity by deriving a randomized communication protocol from 
a randomized \SWA. This is in fact a standard technique for obtaining lower bounds for streaming algorithms.
In the next section, we present the necessary background from communication complexity; see \cite{KuNi97}
for a detailed introduction.  

\subsection{Communication complexity}

We need a promise version of randomized one-way communication complexity; see also \cite{KremerNR99}.
Consider a function $f \colon D \to \{0,1\}$ where $D \subseteq X \times Y$ for some finite sets $X$ and $Y$.
A {\em randomized one-way (communication) protocol} $P=(a,b)$ consists of  functions
$a \colon X\times R_a \to \{0,1\}^*$ and $b \colon \{0,1\}^* \times Y \times R_b \to\{0,1\}$, where
$R_a$ (resp., $R_b$) is the finite set of random choices of Alice (resp., Bob).
The {\em cost} of $P$ is the maximum number of bits transmitted by Alice, i.e.
\[
	\mathrm{cost}(P) = \max_{x \in X, r_a \in R_a} |a(x,r_a)|. 
\]
Moreover, probability distributions are given on $R_a$ (resp., $R_b$).
Alice computes from her input $x \in X$ and a random choice $r_a \in R_a$
the value $a(x,r_a)$ and sends it to Bob.
Using this value, his input $y \in Y$ and a random choice $r_b \in R_b$
he outputs $b(a(x,r_a),y,r_b)$. 
The random choices $r_a \in R_a, r_b \in R_b$ are chosen independently from their respective
distributions.
The protocol $P$ {\em computes} $f$ if 
for all $(x,y) \in D$ we have
\begin{equation} \label{eq-rcc}
\Pr_{r_a \in R_a, r_b \in R_b}[P(x,y) \neq f(x,y)] \le \frac{1}{3}.
\end{equation}
where $P(x,y)$ is the random variable $b(a(x,r_a),y,r_b)$.
Note that for $(x,y) \not\in D$, there is no requirement for the probability that Bob outputs
$1$; for instance, this probability can be $1/2$.

A deterministic one-way protocol with cost $s$ is a randomized one-way protocol with cost $s$, where $R_a$ and $R_b$ 
are singleton sets.

The (worst case) {\em randomized one-way communication complexity} $C(f)$ of $f$ is the minimal cost among all one-way randomized protocols that compute $f$ (with an arbitary number of random bits).
The choice of the constant $\frac{1}{3}$ in \eqref{eq-rcc} is arbitrary in the sense that changing the constant
to any $\eps < 1/2$ only changes the cost $C(f)$ by a fixed constant (depending on $\eps$), see \cite[p.~30]{KuNi97}.

In this paper we will use established lower bounds on the following functions
for $n \ge 1$:
\begin{itemize}
\item the {\em index function} $\mathrm{IDX}_n \colon \{0,1\}^n\times \{1,\dots,n\} \to \{0,1\}$ where
\begin{align*}
	\mathrm{IDX}_n(a_1 \cdots a_n,i) = a_i .
\end{align*}

\item the {\em greater-than function} $\mathrm{GT}_n \colon \{1,\dots,n\}\times \{1,\dots,n\}\to \{0,1\}$ where
\begin{align*}
&\mathrm{GT}_n(i,j)= \begin{cases} 1, & \text{if $i>j$}, \\ 0, & \text{if $i \le j$.} \end{cases}
\end{align*}

\item the {\em equality function} $\mathrm{EQ}_n \colon \{1,\dots,n\}\times \{1,\dots,n\}\to \{0,1\}$ where
\begin{align*}
&\mathrm{EQ}_n(i,j)= \begin{cases} 1, & \text{if $i=j$}, \\ 0, & \text{if $i \neq j$.} \end{cases}
\end{align*}
\end{itemize}

\begin{theorem} \label{thm:coco} The following hold:
\begin{enumerate}
\item\label{index} $C(\mathrm{IDX}_n) \in \Theta(n)$ \cite[Theorem~3.7]{KremerNR99}
\item\label{gt} $C(\mathrm{GT}_n) \in \Theta(\log n)$ \cite[Theorem~3.8]{KremerNR99}
\item\label{eq} $C(\mathrm{EQ}_n) \in \Theta(\log\log n)$  \cite{KuNi97} 
\end{enumerate}
\end{theorem}

\begin{remark}
Usually, the equality function is defined as the function 
$\mathrm{EQ}'_n \colon \{0,1\}^n \times \{0,1\}^n\to \{0,1\}$ with 
$\mathrm{EQ}'_n(u,v) = 1$ if and only if $u = v$ for all $u,v \in \{0,1\}^n$.
It is well known that $C(\mathrm{EQ}'_n) \in \Theta(\log n)$ which implies
$C(\mathrm{EQ}_n) \in \Theta(\log\log n)$.
\end{remark}
We need a promise version of point (1) from Theorem~\ref{thm:coco}.
In the following lemma we consider the restriction 
$\mathrm{IDX}_n|_D \colon D \to \{0,1\}$ for a subset $D \subseteq \{0,1\}^n \times \{1, \dots, n\}$.
\begin{lemma}[promise version of IDX]
	\label{lem:promise-idx}
	Let $f = \mathrm{IDX}_n|_D$ be a promise version of $\mathrm{IDX}_n$
	where $D \subseteq \{0,1\}^n \times \{1, \dots, n\}$ satisfies
	$|\{ i \in [1, n] : (x,i) \in D \}| \ge \frac{7}{8} n$ for all $x \in \{0,1\}^n$.
	Then $C(f) \in \Omega(n)$ holds.
\end{lemma}

\begin{proof}
        Let $P' = (a,b)$ be a randomized one-way protocol, which computes 
        $f$ and has error probability $\leq 1/8$.
        Let $R_a$ (resp., $R_b$) be the set of random choices of Alice (resp., Bob).
        For $(x,i) \in D$
        consider the $\{0,1\}$-valued random variable 
        \[ X_{x,i} = |b(a(x,r_a),i,r_b) - \mathrm{IDX}_n(x,i)| . \]
        Its expectation is the error probability. Hence, for all $(x,i) \in D$
        we have
        \[
        E[X_{x,i}] \leq \frac{1}{8} .
        \]
        By taking the average over all $(x,i) \in D$ and using linearity of expectation, we obtain
	\[
		E\bigg[\frac{1}{|D|} \sum_{(x,i) \in D} X_{x,i}\bigg] = 
		\frac{1}{|D|} \sum_{(x,i) \in D} E[X_{x,i}] \le \frac{1}{8}.
	\]
	Hence, there exist random choices $r_a \in R_a, r_b \in R_b$ such that
         \[
		\frac{1}{|D|} \sum_{(x,i) \in D}|b(a(x,r_a),i,r_b) - \mathrm{IDX}_n(x,i)| \le \frac{1}{8}.
	\]
        	Let $P = (a(\cdot,r_a), b(\cdot,\cdot,r_b))$ be the deterministic one-way protocol obtained from $P'$ by fixing Alice's (resp., Bob's) random
	choice to $r_a$ (resp., $r_b$). 
	We then have	
	\begin{equation}
		\Pr_{(x,i) \in D}[P(x,i) \neq \mathrm{IDX}_n(x,i)] \le \frac{1}{8}
		\label{eq:hard-dist}
	\end{equation}
	where we consider the uniform distribution on $D$. The above argument is of course nothing else than
	the easy direction of Yao's min-max principle (see also  \cite[Theorem~3.20]{KuNi97}).
	
	Let $\nu \colon \{0,1\}^n \to \{0,1\}^n$
	be the function $\nu(x) = P(x,1) P(x,2) \cdots P(x,n)$
	where we interpret $P(x,i) = 0$ for all $(x,i) \notin D$.
	If $s$ is the cost of $P$ (which is equal to the cost of $P'$),
	then $V = \{ \nu(x) \colon x \in \{0,1\}^n \}$ has size at most $2^s$.
	Let $\Delta(x,y)$ denote the Hamming distance between two words $x,y \in \{0,1\}^n$,
	i.e., the number of positions where $x$ and $y$ differ.
	For each $x = x_1 x_2 \cdots x_n \in \{0,1\}^n$ we have
	\[
		\Delta(x,\nu(x)) =  \sum_{i=1}^n |P(x,i)-x_i|
 \le \sum_{i : (x,i) \in D} |P(x,i)-x_i| + \frac{n}{8}.
	\]
	By summing over all words $x \in \{0,1\}^n$ we obtain
	\begin{align*}
		\sum_{x \in \{0,1\}^n} \Delta(x,\nu(x)) &\le \sum_{(x,i) \in D} |P(x,i)-x_i| + \frac{2^n \cdot n}{8}\\
		& \stackrel{\eqref{eq:hard-dist}}{\le} \frac{|D|}{8} + \frac{2^n \cdot n}{8} \\ &\le \frac{2^n \cdot n}{8} + \frac{2^n \cdot n}{8} \\
		& = \frac{2^n \cdot n}{4}.
	\end{align*}
	This implies that the expected Hamming distance $\Delta(x,\nu(x))$ for a randomly picked word $x \in \{0,1\}^n$
	(under the uniform distribution on $\{0,1\}^n$) is bounded by
	\[
		E_{x \in \{0,1\}^n}[\Delta(x,\nu(x))] \le \frac{n}{4}.
	\]
	Applying Markov's inequality we get
	\[
		\Pr_{x \in \{0,1\}^n} [\Delta(x,\nu(x)) \ge \frac{3n}{8}] \le \frac{n}{4} \cdot \frac{8}{3n} = \frac{2}{3}
	\]
	and therefore
	\[
		\Pr_{x \in \{0,1\}^n} [\Delta(x,\nu(x)) \le \frac{3n}{8}] \ge \Pr_{x \in \{0,1\}^n} [\Delta(x,\nu(x)) < \frac{3n}{8}] \ge \frac{1}{3}.
	\]
	This means that there exists a set $U \subseteq \{0,1\}^n$ of size at least $\frac{2^n}{3}$
	such that for each $x \in U$ there exists a word $y \in V$ (namely $\nu(x)$) with $\Delta(x,y) \le \frac{3n}{8}$.
	Denote by $B_r(y) = \{ x \in \{0,1\}^n \colon \Delta(x,y) \le r \}$
	the ball of radius $r$. It is known that
	\[
		|B_{\eps n}(y)| = \sum_{i=0}^{\lfloor \eps n \rfloor} {n \choose i} \le 2^{H(\eps)n}
	\]
	where $0 < \eps < \frac{1}{2}$ and $H(\eps) = -\eps \log_2 \eps - (1-\eps) \log_2 (1-\eps)$
	is the {\em binary entropy function} \cite[Lemma~2.3.5]{Gray90}.
	Since
	\[
		U \subseteq \bigcup_{y \in V} B_{\frac{3n}{8}}(y)
	\]
	we know that $\frac{2^n}{3} \le |U| \le |V| \cdot |B_{\frac{3n}{8}}|$
	and therefore
	\[
		2^s \ge |V| \ge \frac{1}{3} \cdot 2^{(1-H(3/8))n}.
	\]
	Since $0 < H(3/8) \approx  0.954 < 1$ we conclude $s = \Omega(n)$.
\end{proof}

\subsection{Randomized lower bounds for failure ratio zero}

In this section  we prove the lower bounds in \eqref{lower-loglog}, \eqref{lower-log} and \eqref{lower-lin} from Theorem~\ref{thm:quatrochotomy}. 

\subsubsection{Linear lower bound}

We start with the proof of \eqref{lower-lin} from Theorem~\ref{thm:quatrochotomy}, which extends our linear
space lower bound from the deterministic setting \cite{GHL16} to the randomized setting.
We will need the following property:
\begin{lemma}\label{lemma:linearprop}
If $L \in \Reg \setminus \langle \LI, \Len \rangle$,
then there are $x_0,u_0,x_1,u_1,u\in\Sigma^*$ with the following properties:
\begin{itemize}
\item $|x_0|=|x_1| \ge 1$ and $|u_0|=|u_1| \ge 1$,
\item $u_0\{x_0u_0,x_1u_1\}^*u\;\cap L = \emptyset$ and
\item $u_1\{x_0u_0,x_1u_1\}^*u\;\subseteq L$.
\end{itemize}
\end{lemma}
\begin{proof}
By \cite{GHKLM18} the class of regular languages which admit a deterministic \SWA\ using space $\mathcal{O}(\log n)$ is exactly $\langle \LI, \Len \rangle$. Using this characterization, the words with the properties above
are constructed in \cite[proof of Theorem~9]{GHL16}. 
\end{proof}

\begin{theorem}
If $L \in \Reg \setminus \langle \LI, \Len \rangle$,
then every randomized \SWA\ $\R$ for $L$ satisfies $f(\R,n) \notin o(n)$.
\end{theorem}

\begin{proof}
Consider a randomized \SWA\ $\R=(R_n)_{n \ge 0}$ for $L$. 
Consider the words $x_0,u_0,x_1,u_1,u\in\Sigma^*$ described in Lemma~\ref{lemma:linearprop}.
Note that the length of all these words only depends on the language $L$ and therefore is independent from the window length.
Let $m \ge 1$ be an arbitrary integer.
Using this information, we describe a randomized one-way communication protocol for $\mathrm{IDX}_m$.

Let $\alpha = \alpha_1 \cdots \alpha_m \in\{0,1\}^m$ be Alice's input and $i\in\{1,\dots,m\}$ be Bob's input.
Let $n =|x_0u_0|\cdot m+|u_0|+|u|\in\Theta(m)$.
Alice and Bob use their random choices in order to simulate the random choices
of the probabilistic automaton $R_{n}$ on certain words.
We define the word
\[
	w_\alpha = x_{\alpha_1} u_{\alpha_1} x_{\alpha_2} u_{\alpha_2} \cdots x_{\alpha_m} u_{\alpha_m}.
\]
Since $n=|x_0u_0|\cdot m+|u_0|+|u|$, $|x_0|=|x_1|$ and $|u_0|=|u_1|$ we have
\[
	\last_n(w_\alpha (x_0u_0)^i u) = u_{\alpha_i} x_{\alpha_{i+1}} u_{\alpha_{i+1}} \cdots x_{\alpha_m} u_{\alpha_m} (x_0u_0)^{i}u,
\]
which belongs to $L$ if and only if $\alpha_i = 1$.
This results in the following protocol $P_m$ for $\mathrm{IDX}_m$:
Alice simulates $R_n$ on $w_\alpha$ and sends the reached state to Bob.
If $1 \le i \le m$ is Bob's input, he continues the run in $R_n$ with the word $(x_0u_0)^i u$.
The algorithm $R_n$ then accepts (resp., rejects) with probability $2/3$ if $\alpha_i = 1$ (resp., $\alpha_i = 0$).
The cost of the protocol $P_m$ is bounded by $f(\R,n)$ (the maximal encoding length of reachable states in $R_n$).
By Theorem~\ref{thm:coco} we have
\[
	f(\R,|x_0u_0|\cdot m+|u_0|+|u|) = f(\R,n) \ge \mathrm{cost}(P_m) \in \Omega(m)
\]
and therefore $f(\R,n) \notin o(n)$.
\end{proof}

\subsubsection{Logarithmic lower bound}

Next we prove point \eqref{lower-log} from Theorem~\ref{thm:quatrochotomy}.
For that, we need the following automaton property.
Let $A = (Q,\Sigma,q_0,\delta,F)$ be a DFA in the following. 
A state $q$ is {\em trivial} if $\delta(q,x) \neq q$ for all $x \in \Sigma^+$, otherwise it is {\em non-trivial}.

\begin{definition}[synchronized state pair]\label{def:looping}
A pair $(p,q)\in Q\times Q$ of states is called {\em synchronized}
if there exist words $x,y,z\in\Sigma^*$ with $|x|=|y|=|z| \geq 1$ such that $\delta(p,x)=p$, $\delta(p,y)=q$ and $\delta(q,z)=q$.
A pair $(p,q)$ is called {\em reachable} from a state $r$ if $p$ is reachable from $r$.
A state pair $(p,q)$ is called {\em $F$-consistent} if either $\{p,q\} \cap F = \emptyset$.
or $\{p,q\} \subseteq F$.
\end{definition}

We remark that synchronized state pairs have no connection to the notion of synchronizing words.

\begin{lemma}
	\label{lem:factorial-length}
	A state pair $(p,q)$ is synchronized
	if and only if $p$ and $q$ are non-trivial
	and there exists a word $y \in \Sigma^+$ such that $|Q|!$ divides $|y|$ and $\delta(p,y) = q$
\end{lemma}

\begin{proof}
	Let $x,y,z\in\Sigma^+$ with $|x|=|y|=|z|=k$ such that $\delta(p,x)=p$, $\delta(p,y)=q$, and $\delta(q,z)=q$.
	Then $p$ and $q$ are non-trivial and we have $\delta(p,x^{|Q|!-1}y) = q$
	where $x^{|Q|!-1}y$ has length $(|Q|!-1) \cdot k + k = |Q|! \cdot k$.

	Conversely, assume that $p$ and $q$ are non-trivial
	and there exists a word $y \in \Sigma^+$ such that $|Q|!$ divides $|y|$ and $\delta(p,y) = q$.
	Since the states $p$ and $q$ are non-trivial,
	there are words $x$ and $z$ of length at most $|Q|$ with $\delta(p,x) = p$ and $\delta(q,z) = q$.
	These words can be pumped up to have length $|y|$.
\end{proof}

Let $Q = T \cup N$ be the partition of the state set into the set $T$ of trivial states and the set $N$ of non-trivial states.
A function $\beta \colon \mathbb{N} \to \{0,1\}$ is {\em $k$-periodic} if $\beta(i) = \beta(i+k)$ for all $i \in \mathbb{N}$.

\begin{lemma}
	\label{lem:periodic}
	Assume that every synchronized pair in $A$ which is reachable from $q_0$ is $F$-consistent.
	Then for every word $v \in \Sigma^*$ of length at least $|Q|! \cdot (|T|+1)$
	there exists a $|Q|!$-periodic function $\beta_v \colon \mathbb{N} \to \{0,1\}$ such that the following holds:
	If $w \in v\Sigma^*$ and $\delta(q_0,w) \in N$, then we have
	$w \in L$ iff $\beta(|w|) = 1$.
\end{lemma}

\begin{proof}
	Let $v = a_1 a_2\cdots a_k$  with
	$k \ge |Q|! \cdot (|T|+1)$, and consider the run
	\begin{equation}
		\label{eq:prefix-run}
		q_0 \xrightarrow{a_1} q_1 \xrightarrow{a_2} \cdots \xrightarrow{a_k} q_k
	\end{equation}
	of $A$ on $v$.
	Clearly, each trivial state can occur at most once in the run.
	First notice that for each $0 \le i \le |Q|!-1$ at least one of the states in
	\[
		Q_i = \{ q_{i+j|Q|!} : 0 \le j \le |T| \}
	\]
	is non-trivial because otherwise the set would contain $|T|+1$ pairwise distinct trivial states.
	Furthermore, we claim that the non-trivial states in $Q_i$ are either all final or all non-final:
	Take two non-trivial states $q_{i+j_1|Q|!}$ and $q_{i+j_2|Q|!}$ with $j_1 < j_2$.
	Since we have a run of length $(j_2-j_1)|Q|!$ from $q_{i+j_1|Q|!}$ to $q_{i+j_2|Q|!}$,
	the states form a synchronized pair by Lemma~\ref{lem:factorial-length}.
	Hence, by assumption the two states are $F$-consistent.
	
	Now define $\beta_v \colon \mathbb{N} \to \{0,1\}$ by
	\[
		\beta_v(m) = \begin{cases}
		1 & \text{if the states in } Q_{m \bmod |Q|!} \cap N \text{ are final}, \\
		0 & \text{if the states in } Q_{m \bmod |Q|!} \cap N \text{ are non-final},
		\end{cases}
	\]
	which is well-defined by the remarks above.
	Clearly $\beta_v$ is $|Q|!$-periodic.
	
	Let $w = a_1 \cdots a_m \in v\Sigma^*$ be a word of length $m \ge k$.
	The run of $A$ on $w$ prolongs the run in \eqref{eq:prefix-run}:
	\[
		q_0 \xrightarrow{a_1} q_1 \xrightarrow{a_2} \cdots \xrightarrow{a_m} q_m.
	\]
	Assume that $q_m \in N$.
	As argued above, there is a position $0 \le i \le k$
	such that $i \equiv m \pmod{|Q|!}$ and $q_i \in N$.
	Hence $(q_i,q_m)$ is a synchronized pair by Lemma~\ref{lem:factorial-length}
	which is $F$-consistent by assumption.
	Therefore $w \in L$ iff $q_m \in F$ iff $q_i \in F$ iff $\beta_v(|w|) = 1$.
\end{proof}

\begin{lemma}\label{lem:sync-pair}
Assume that every synchronized pair in $A$ which is reachable from $q_0$ is $F$-consistent.
Then $L(A)$ belongs to $\langle \PT, \PF, \Len \rangle$.
\end{lemma}

\begin{proof}
Let $F_N = N \cap F$ and $F_T = T \cap F$. We decompose $L$ into
\[
	L = L(A, F_N) \cup \bigcup_{q \in F_T} L(A, \{q\}).
\]
First observe that $L(A, \{q\}) \in \PF$ for all $q \in F_T$
because a trivial state $q$ can occur at most once in a run of $A$.

It remains to show  that  $L(A, F_N)$ belongs to $\langle \PT,\PF,\Len \rangle$.
Using the threshold $k=|Q|! \cdot (|T|+1)$,
we distinguish between words of length at most $k-1$ and words of length at least $k$,
and group the latter set by their prefix of length $k$, i.e.,
\[
	L(A, F_N) = (L(A, F_N) \cap \Sigma^{\le k-1})\cup\bigcup_{v\in\Sigma^{k}}(L(A, F_N) \cap v\Sigma^*).
\]
The first part $L(A, F_N) \cap \Sigma^{\le k-1}$ is finite and thus prefix testable.
To finish the proof, we will show that $L(A, F_N) \cap v\Sigma^* \in  \langle \PT, \PF, \Len \rangle$
for each $v\in\Sigma^k$.
Let $v \in \Sigma^k$ and let $\beta_v \colon \mathbb{N} \to \{0,1\}$ be the $|Q|!$-periodic function from Lemma~\ref{lem:periodic}.
We know
\[
	L(A, F_N) \cap v\Sigma^* = (v\Sigma^* \cap \{ w \in \Sigma^* : \beta(|w|) = 1 \}) \setminus L(A, T).
\]
The language $\{ w \in \Sigma^* : \beta(|w|) = 1 \}$ is 
a regular length language, $v \Sigma^*$ is prefix testable and $L(A, T) \in \langle \PF \rangle$.
\end{proof}

The following lemma is an immediate consequence of Lemma~\ref{lem:sync-pair}.

\begin{lemma}\label{lemma:logprop-prefix}
If $L \in \Reg \setminus \langle \PT, \PF, \Len \rangle$,
then there exist $u,x,y,z\in\Sigma^*$ with $|x|=|y|=|z| \ge 1$ such that one of the following cases holds:
\begin{itemize}
\item $ux^* \subseteq L\;\text{ and }\;ux^*yz^* \cap L = \emptyset$ 
\item $ux^* \cap L= \emptyset\;\text{ and }\;ux^*yz^* \subseteq L$.
\end{itemize}
\end{lemma}

	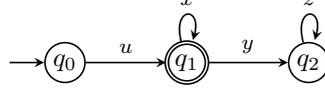
\begin{figure}
			\centering
			\begin{tikzpicture}[semithick,->,>=stealth]
			\node[state, initial left, initial text={}] (q0) {$q_0$};
			\node[state, double, right = 30pt of q0] (p) {$q_1$};
			\node[state, right = 30pt of p] (q) {$q_2$};
			
			\draw (q0) edge node[above] {\footnotesize $u$} (p);
			\draw [loop above] (p) edge node[above] {\footnotesize $x$} (p);
			\draw (p) edge node[above] {\footnotesize $y$} (q);
			\draw [loop above] (q) edge node[above] {\footnotesize $z$} (q);
			\end{tikzpicture}
		\caption{Forbidden pattern for $\langle \PT, \PF, \Len \rangle$ where $|x| = |y| = |z| \ge 1$
		(there is a symmetric case, where $q_1 \notin F$ and $q_2 \in F$).}
		\label{fig:pt-pf-len}
	\end{figure}

Note that $L \in  \langle \PT, \PF, \Len \rangle$ if and only if $L^\rev \in  \langle \ST, \SF, \Len \rangle$.
By applying Lemma~\ref{lemma:logprop-prefix} to $L^\rev$ we get:

\begin{lemma}\label{lemma:logprop}
If $L \in \Reg \setminus \langle \ST, \SF, \Len \rangle$,
then there exist $u,x,y,z\in\Sigma^*$ with $|x|=|y|=|z| \ge 1$ such that one of the following cases holds:
\begin{itemize}
\item $x^*u\subseteq L\;\text{ and }\;z^*yx^*u\cap L = \emptyset$ 
\item $x^*u\cap L=\emptyset\;\text{ and }\;z^*yx^*u\subseteq L$.
\end{itemize}
\end{lemma}

\begin{theorem}\label{lowerlognrandom}
If $L \in \Reg \setminus \langle \ST, \SF, \Len \rangle$,
then every randomized \SWA\ $\R$ for $L$
satisfies $f(\R,n) \notin o(\log n)$.
\end{theorem}

\begin{proof}
Consider the words $u,x,y,z\in\Sigma^*$ described in Lemma~\ref{lemma:logprop}.
Let $\R=(R_n)_{n \ge 0}$ be a randomized \SWA\ for $L$.
Let $m \ge 0$. We describe a randomized one-way protocol $P_m$ for $\mathrm{GT}_m$:
Let $i \in \{1,\dots, m\}$ be the input of Alice and $j\in\{1,\dots, m\}$ be the input of Bob.
Let $n=|x|\cdot m+|u| \in \Theta(m)$.
Alice starts by running the probabilistic automaton $R_n$ on $z^myx^{m-i}$
using her random bits in order to simulate the random choices of $R_n$.
Afterwards, she sends the encoding of the reached state to Bob.
Bob then continues the run of $R_n$ from the transmitted state with the word $x^j u$.
Hence, $R_n$ is simulated on the word $w := z^myx^{m-i} x^{j}u = z^myx^{m-i+j}u$.
We have
\[
	\last_n(w) = \begin{cases}
	z^{i-1-j}yx^{m-i+j}u, & \text{if } i > j, \\
	x^m u, & \text{if } i \le j.
	\end{cases}
\]
By Lemma~\ref{lemma:logprop}, $\last_n(w)$ belongs to $L$ in exactly one of the two cases
$i>j$ and $i \leq j$. Hence Bob can distinguish these two cases with probability at least $2/3$.
It follows that the  protocol computes $\mathrm{GT}_m$ and its cost is bounded by $f(\R,n)$.
By Theorem~\ref{thm:coco} (point~\ref{gt}) we have
\[
	f(\R,|x|\cdot m +|u|) = f(\R,n) \ge \mathrm{cost}(P_m) \in \Omega(\log m)
\]
and therefore $f(\R,n) \notin o(\log n)$.
\end{proof}

\subsubsection{Doubly logarithmic lower bound}

Finally, we need to prove the doubly logarithmic lower bound from point 
\eqref{lower-loglog} in Theorem~\ref{thm:quatrochotomy}.

\begin{lemma}\label{lemma:classC}
Let $A = (Q,\Sigma,q_0,\delta,F)$ be a DFA such that
$\delta(q_0,vx)$ and $\delta(q_0,vy)$ are $F$-consistent
for all $v \in \Sigma^{|Q|}$ and $x,y \in \Sigma^*$ with $|x| = |y|$.
Then $L(A) \in \langle \PT, \Len \rangle$.
\end{lemma}

\begin{proof}
	Let $L = L(A)$.
	We decompose $L$ as
	\[
		L = (L \cap \Sigma^{\le |Q|-1}) \cup \bigcup_{v\in\Sigma^{|Q|}}(L \cap v\Sigma^*).
	\]
	The language $L\cap \Sigma^{\le |Q|-1}$ is finite and thus prefix testable.
	It remains to show that $L \cap v\Sigma^*$ belongs to $\langle \PT, \Len \rangle$ for each $v \in \Sigma^{|Q|}$.
	Consider the regular language
	\[
		v^{-1}L = \{ w \in \Sigma^* : vw \in L \},
	\]
	which is the set of all words which are accepted from $\delta(q_0,v)$.
	Since the states $\delta(q_0,vx)$ and $\delta(q_0,vy)$ are $F$-consistent
	for all words $x,y \in \Sigma^*$ of the same length,
	the language $v^{-1}L$ is a length language. Moreover,
	\[
		L \cap v\Sigma^* = v (v^{-1} L) = 
		(\Sigma^{|Q|} (v^{-1}L)) \cap v\Sigma^*,
	\]
	which belongs to $\langle \PT, \Len \rangle$.
\end{proof}

\begin{lemma}
	\label{lem:pt-len}
	Let $A = (Q,\Sigma,q_0,\delta,F)$ be a DFA for $L$.
	If $L \in \Reg \setminus \langle \PT, \Len \rangle$,
	then there exist words $u,x,y \in \Sigma^*$ and states $q_1, q_2 \in Q$ such that
	$|x| = |y| \geq 1$, $\delta(q_0,u) = q_1$, $\delta(q_1,x) = q_1$, $\delta(q_1,y) = q_2$
	and the pair $(q_1,q_2)$ is not $F$-consistent.
\end{lemma}

\begin{proof}
	By Lemma~\ref{lemma:classC} there exist words $s \in \Sigma^{|Q|}$, $s_0,s_1 \in \Sigma^*$
	with $|s_0| = |s_1|$ such that $\delta(q_0,s s_0) \notin F$ and $\delta(q_0,s s_1)\in F$.
	Since the run of $A$ on $s$ visits $|Q|+1$ states, one of them is visited twice, 
	i.e., there is a factorization $s = vwz$ with $|w|\geq 1$ and a state $p \in Q$ such that
	$\delta(q_0,v) = p = \delta(p,w)$.
	We redefine $s_0 := z s_0$ and $s_1 := z s_1$.
	We can assume that $|w| \ge |s_0|=|s_1|$, otherwise we replace $w$ by $w^i$
	for some integer $i$ with $i \cdot |w| \ge |s_0|$.
	Factorize $w = w_1 w_2$ such that $|w_1| = |s_0| = |s_1|$.
	Let $q_1 := \delta(p,w_1)$.
	If $q_1 \in F$ we choose the words $u=v w_1$, $x=w_2 w_1$ and $y=w_2 s_0$;
	if $q_1 \notin F$ we choose the words $u=v w_1$, $x=w_2 w_1$ and $y=w_2 s_1$.
\end{proof}

	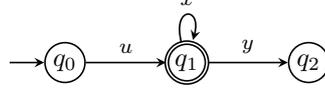
\begin{figure}
			\centering
			\begin{tikzpicture}[semithick,->,>=stealth]
			\node[state, initial left, initial text={}] (q0) {$q_0$};
			\node[state, double, right = 30pt of q0] (p) {$q_1$};
			\node[state, right = 30pt of p] (q) {$q_2$};
			
			\draw (q0) edge node[above] {\footnotesize $u$} (p);
			\draw [loop above] (p) edge node[above] {\footnotesize $x$} (p);
			\draw (p) edge node[above] {\footnotesize $y$} (q);
			\end{tikzpicture}
		\caption{Forbidden pattern for $\langle \PT, \Len \rangle$ where $|x| = |y| \ge 1$
		(there is a symmetric case, where $q_1 \notin F$ and $q_2 \in F$).}
		\label{fig:pt-len}
	\end{figure}

\begin{lemma}\label{lemma:loglogprop-prefix}
	If $L \in \Reg \setminus \langle \PT, \Len \rangle$, then
	there exist words $u,x,y,z \in \Sigma^*$ with $|x|=|y|=|z| \ge 1$
	such that one of the following cases holds:
	\begin{itemize}
		\item $ux^* \subseteq L$\; and \;$ux^*yz^* \cap L = \emptyset$ 
		\item $ux^* \cap L= \emptyset$\; and \;$ux^*yz^* \subseteq L$
		\item $u x^*\subseteq L$\; and \; $ux^*y\cap L = \emptyset$\; and \;$ux^*y z^+\subseteq L$
		\item $u x^*\cap L=\emptyset$\; and \;$ux^*y\subseteq L$\; and \;$ux^*y z^+ \cap L =\emptyset$
	\end{itemize} 
\end{lemma}

\begin{proof}
	Let $A= (Q,\Sigma,q_0,\delta,F)$ be a DFA for $L$.
	Let $u,x,y$ be the words and $q_1,q_2$ the states from Lemma~\ref{lem:pt-len}.
	Without loss of generality we assume $q_1 \in F$ and $q_2 \notin F$. 
	Consider the state sequence
	\[
		q_2 = p_0 \xrightarrow{y} p_1 \xrightarrow{y} p_2 \xrightarrow{y} \cdots
	\]
	which is ultimately periodic, i.e. there exists $t \ge 0$ and $d \ge 1$ such that
	$p_i = p_{i+d}$ for all $i \ge t$.
	Let $k = (t+1)d$ and define $q_3 = p_k$.
	Then we have $\delta(q_2,y^k) = q_3 = \delta(q_3,y^k)$.
	To summarize, we have the following run in $A$:
	\[
		q_0 \xrightarrow{u} q_1 \xrightarrow{x^k} q_1 \xrightarrow{x^{k-1} y} q_2 \xrightarrow{y^k} q_3 \xrightarrow{y^k} q_3
	\]
	We have $u (x^k)^* \subseteq L$.
	If $q_3 \notin F$ then $u (x^k)^* (x^{k-1} y)(y^k)^* \cap L = \emptyset$, and the first case from the lemma holds.
	If $q_3 \in F$ then $u (x^k)^* (x^{k-1} y) \cap L = \emptyset$ 
	and $u (x^k)^* (x^{k-1} y)(y^k)^+ \subseteq L$, and the third case from the lemma holds.
\end{proof}

Lemma~\ref{lemma:loglogprop-prefix} applied to $L^\rev$ yields:

\begin{lemma}\label{lemma:loglogprop}
	If $L \in \Reg \setminus \langle \ST, \Len \rangle$,
	there exist words $u,x,y,z \in \Sigma^*$ with $|x|=|y|=|z| \ge 1$
	such that one of the following cases holds:
	\begin{itemize}
		\item $x^*u \subseteq L$\; and \;$z^*yx^*u \cap L = \emptyset$ 
		\item $x^*u \cap L= \emptyset$\; and \;$z^*yx^*u \subseteq L$
		\item $x^*u \subseteq L$\; and \; $yx^*u\cap L = \emptyset$\; and \;$z^+yx^*u\subseteq L$
		\item $x^*u \cap L=\emptyset$\; and \;$yx^*u\subseteq L$\; and \;$z^+yx^*u \cap L =\emptyset$
	\end{itemize} 
\end{lemma}

\begin{theorem}
If $L \in \Reg \setminus \langle \ST, \Len \rangle$,
then every randomized \SWA\ $\R$ for $L$ satisfies $f(\R,n) \notin o(\log\log n)$.
\end{theorem}

\begin{proof}
We apply Lemma~\ref{lemma:loglogprop} to the language $L$.
If $L$ satisfies one of the first two conditions (which are the conditions from Lemma~\ref{lemma:logprop})
then by the argument from the proof of Theorem~\ref{lowerlognrandom} we get $f(\R,n)\notin o(\log n)$. 

Let us assume that the fourth condition holds (the third is analogous).
Let $\R=(R_n)_{n \ge 0}$ be a randomized \SWA\ for $L$.
Let $m \ge 1$.
We describe a randomized one-way protocol $P_m$ for $\mathrm{EQ}_m$:
Let $i \in \{1,\dots, m\}$ be the input of Alice and $j \in \{1,\dots, m\}$ be the input of Bob.
Let $n=|x|\cdot (m+1) +|u| \in \Theta(m)$.
Alice now runs the probabilistic automaton $R_n$ on $z^m y x^i$.
Afterwards, she sends the encoding of the reached state to Bob.
Bob continues the run of $R_n$ from the transmitted state with the word $x^{m-j} u$
and outputs $1$ iff $R_n$ accepts.
Hence, $R_n$ is simulated on the word $z^m y x^{i-j+m} u$. We have
\[
	\last_n(z^m y x^{i-j+m} u) =
	\begin{cases}
	z^{j-i} yx^{i-j+m} u, & \text{if } i < j, \\
	yx^m u, & \text{if } i = j, \\
	x^{m+1} u, & \text{if } i > j.
	\end{cases}
\]
Therefore $\last_n(z^m y x^{i-j+m} u) \in L$ iff $i = j$.
It follows that this protocol computes $\mathrm{EQ}_m$.
By Theorem~\ref{thm:coco} (point~\ref{eq}) we have
\[
	f(\R,|x|\cdot (m+1) +|u|) = f(\R,n) \ge \mathrm{cost}(P_m) \in \Omega(\log \log m)
\]
and therefore $f(\R,n) \notin o(\log \log n)$.
\end{proof}

\subsection{Randomized lower bounds for arbitrarily small failure ratio}

In this section we prove the lower bound \eqref{upper-O(1)-failure} 
from Theorem~\ref{thm:dichotomy}.

Viewing a DFA $A$ as a directed graph with the vertex set $Q$ and arcs $(q,\delta(q,a))$
for all $q \in Q$, $a \in \Sigma$, we can talk about strongly connected components of $A$.
In other words, a {\em strongly connected component (SCC)} of $A$ is an inclusion maximal set of states $S \subseteq Q$
such that for all $p,q \in S$ there exists a word $x \in \Sigma^*$ with $\delta(p,x) = q$.
An SCC $S$ is {\em maximal} if for all $q \in S$ and $x \in \Sigma^*$ we have $\delta(q,x) \in S$.
An SCC $S$ is {\em trivial} if $S = \{q\}$ for a trivial state $q$.

\begin{definition}
Let $A = (Q,\Sigma,q_0,\delta,F)$ be a DFA.
\begin{itemize}
\item An SCC $C \subseteq Q$ is {\em well-behaved} if for all $q \in C$
	and $u,v \in \Sigma^*$ with $|u|=|v|$ and $\delta(q,u), \delta(q,v) \in C$ we have:
	$\delta(q,u) \in F$ if and only if $\delta(q,v) \in F$.
\item $A$ is {\em well-behaved} if every SCC which is reachable from $q_0$ is well-behaved. 
\item A state $q \in Q$ is called {\em positively idempotent} if there exists a word $x \in \Sigma^+$
such that $\delta(q_0,x) = \delta(q,x) = q$.
\item $A$ is {\em idempotently well-behaved} if every SCC which is reachable from a positively idempotent state
is well-behaved.
\end{itemize}
\end{definition}

\begin{theorem}[\cite{GHKLM18}] \label{thm:wb}
A DFA recognizes a language in $\langle \RI, \Len \rangle$ if and only if
it is well-behaved.
\end{theorem}

\begin{lemma}
	\label{lem:before-idempotent}
	Let $A= (Q,\Sigma,q_0,\delta,F)$ be a DFA in which all states which are reachable from positively idempotent states
	are non-final. Then $L = L(A)$ is a finite union of regular suffix-free languages.
\end{lemma}

\begin{proof}
	Let $A' = (Q',\Sigma,q_0',\delta',F')$ be a DFA for $L^\rev$.
	We have
	\[
		L^\rev = \bigcup_{q \in F'} L(A',\{q\}).
	\]
	Let $q \in F'$.
	We claim that for every state $q \in F'$ the language $L(A',\{q\})$ is prefix-free.
	Assume that $x,xy \in L(A',\{q\})$ with $y \neq \varepsilon$, i.e. $\delta(q_0',x) = \delta(q_0',xy) = q$
	and hence also $\delta(q_0',xy^i) = q$ for all $i \ge 0$.
	Now consider the run of $A$
	\[
		q_0 \xrightarrow{y^\rev} q_1 \xrightarrow{y^\rev} q_2 \xrightarrow{y^\rev} \cdots
	\]
	which is ultimately periodic.
	Let $t \ge 0$ and $d \ge 1$ be such that $q_i = q_{i+d}$ for all $i \ge t$.
	For $k = (t+1)d$ the state $q_k$ is positively idempotent because $\delta(q_0,(y^\rev)^k) = q_k = \delta(q_k,(y^\rev)^k)$.
	However, $xy^k \in L^\rev$ and therefore $(y^\rev)^kx^\rev \in L$,
	which means that a final state is reachable from $q_k$, contradiction.
\end{proof}

\begin{theorem} \label{thm-IWB}
Let $A$ be an idempotently well-behaved DFA.
Then $L = L(A)$ belongs to $\langle \RI, \SF, \Len \rangle$. 
\end{theorem}

\begin{proof}	
	Consider an idempotently well-behaved DFA $A= (Q,\Sigma,q_0,\delta,F)$ for $L$.
	Let $F_{\mathrm{id}} \subseteq F$ be the set of final states which
	are reachable from a positively idempotent state.
	The DFA $B = (Q,\Sigma,q_0,\delta,F_{\mathrm{id}})$ is well-behaved. Theorem~\ref{thm:wb}
	implies that $L(B) \in \langle \RI, \Len \rangle$.
	In the DFA $C = (Q,\Sigma,q_0,\delta,F \setminus F_{\mathrm{id}})$ all final states are not reachable
	from positively idempotent states.
	Lemma~\ref{lem:before-idempotent} implies that $L(C) \in \langle \SF \rangle$. 
	Hence $L = L(B) \cup L(C) \in \langle \RI, \SF, \Len \rangle$.
\end{proof}

	\begin{figure}
			\centering
			\begin{tikzpicture}[semithick,->,>=stealth]
			\node[state, initial left, initial text={}] (q0) {$q_0$};
			\node[state, right = 30pt of q0] (p) {$p$};
			\node[state, right = 30pt of p] (q) {$q$};
			\node[state, above right = 25pt and 30pt of q] (qn) {};
			\node[state, double, below right = 25pt and 30pt of q] (qp) {};

			\draw (q0) edge node[above] {\footnotesize $x$} (p);
			\draw [loop above] (p) edge node[above] {\footnotesize $x$} (p);
			\draw (p) edge node[above] {\footnotesize $u$} (q);
			\draw [bend left = 15] (q) edge node[above left] {\footnotesize $y_0$} (qn);
			\draw [bend left = 15] (qn) edge node[below right] {\footnotesize $z_0$} (q);
			\draw [bend right = 15] (q) edge node[below left] {\footnotesize $y_1$} (qp);
			\draw [bend right = 15] (qp) edge node[above right] {\footnotesize $z_1$} (q);
			\end{tikzpicture}
		\caption{Forbidden pattern for $\langle \RI, \SF, \Len \rangle$ where $|x| \ge 1$ and $|y_0| = |y_1|$.}
		\label{fig:iwb}
	\end{figure}
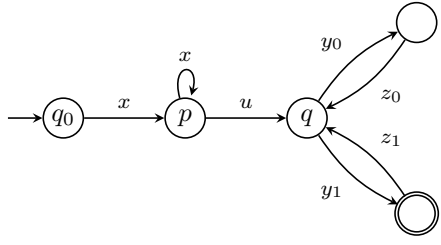

\begin{lemma} \label{lemma-RI-SF-Len}
If $L \in \Reg \setminus \langle \RI, \SF, \Len \rangle$, then there exist words $x,y_0,y_1,z_0,z_1 \in \Sigma^+, u \in \Sigma^*$
such that
\begin{itemize}
\item $|y_0|=|y_1|$, $|x| = |y_0z_0| = |y_1z_1|$,
\item $\{x\}^+ u \{y_0z_0\}^* y_0 \cap L = \emptyset$,
\item $\{x\}^+ u \{y_1z_1\}^* y_1 \subseteq L$.
\end{itemize}
\end{lemma}

\begin{proof}
Let $A= (Q,\Sigma,q_0,\delta,F)$ be a DFA for $L$.
By Theorem~\ref{thm-IWB} $A$ is not idempotently well-behaved.
Hence, $A$ contains
a positively idempotent state $p$ and an SCC that is reachable from $p$ and
not well-behaved. Therefore, there exist words $x,y_0,y_1,z_0,z_1 \in \Sigma^+, u \in \Sigma^*$
	where $|y_0|=|y_1|$
	and a state $q \in Q$ such that
	\begin{itemize}
		\item $\delta(q_0, x) = p = \delta(p, x)$,
		\item $\delta(p,u) = q = \delta(q,y_0z_0) = \delta(q,y_1z_1)$,
		\item $\delta(q,y_0) \notin F$ and $\delta(q,y_1) \in F$.
	\end{itemize}
	This implies $\{x\}^+ u \{y_0z_0\}^* y_0 \cap L = \emptyset$  and $\{x\}^+ u \{y_1z_1\}^* y_1 \subseteq L$.
	Finally, we can ensure that $|x| = |y_0z_0| = |y_1z_1|$ by
	replacing $z_0$ by $(z_0y_0)z_0^{|x| \cdot |y_1z_1|-1}$, $z_1$ by $(z_1y_1)z_1^{|x| \cdot |y_0z_0|-1}$
	and $x$ by $x^{|y_0z_0| \cdot |y_1z_1|}$.
\end{proof}
Using Lemma~\ref{lemma-RI-SF-Len} for the language $L^\rev$ yields:

\begin{lemma} \label{lemma-LI-PF-Len}
If $L \in \Reg \setminus \langle \LI, \PF, \Len \rangle$, then  there exist words $x,y_0,y_1,z_0,z_1 \in \Sigma^+, u \in \Sigma^*$
such that
\begin{itemize}
\item $|y_0|=|y_1|$, $|x| = |z_0y_0| = |z_1y_1|$,
\item $y_0 \{z_0y_0\}^* u \{x\}^+ \cap L = \emptyset$,
\item $y_1 \{z_1y_1\}^* u \{x\}^+ \subseteq L$.
\end{itemize}
\end{lemma}

\begin{theorem}
	Let $L \in \Reg \setminus \langle \LI, \PF, \Len \rangle$.
	Then there exists $0 < \phi < 1$ such that
	any randomized \SWA\ $\R$ for $L$ 
	satisfies $f(\R,n) \notin o(n)$.
\end{theorem}

\begin{proof}
	By Lemma~\ref{lemma-LI-PF-Len} there exist words $x,y_0,y_1,z_0,z_1 \in \Sigma^+, u \in \Sigma^*$
         such that
         \begin{itemize}
         \item $|y_0|=|y_1|$, $k:= |x| = |z_0y_0| = |z_1y_1|$,
         \item $y_0 \{z_0y_0\}^* u \{x\}^+ \cap L = \emptyset$,
         \item $y_1 \{z_1y_1\}^* u \{x\}^+ \subseteq L$.
         \end{itemize}
	For each $m \ge 1$ and $\alpha = \alpha_1 \cdots \alpha_m \in \{0,1\}^m$ define the
	word
	\[
		w_\alpha = z_{\alpha_1}y_{\alpha_1} \cdots z_{\alpha_m}y_{\alpha_m},
	\]
	of length $km$. Since $|w_\alpha u x^m| = \mathcal{O}(m)$ there exist numbers $c > 0$ and $m_0 \ge 0$ such that for all $m \ge m_0$ 
	and $\alpha \in \{0,1\}^m$ we have
	$N := |w_\alpha u x^m|+1 = 2km + |u|+1 \le cm$.
	
	Let $\phi = \frac{1}{8c}$ and
         suppose that $\R = (R_n)_{n \ge 0}$ is a randomized SWA for $L$.
         Let $m \ge m_0$ and
	choose the window length $n = km + |u| + |y_0| = km + |u| + |y_1| \in \Theta(m)$.
	Notice that for all $\alpha \in \{0,1\}^m$ and $i \in [1,m]$ we have
	\[
		\last_n(w_\alpha u x^i) = y_{\alpha_i} z_{\alpha_{i+1}}y_{\alpha_{i+1}} \cdots z_{\alpha_m}y_{\alpha_m} u x^i
	\]
	and this window belongs to $L$ if and only if $\alpha_i = 1$.
	For each $\alpha \in \{0,1\}^m$ we define the set
	\[
		I(\alpha) = \{ i \in [1,m] : \eps(R_n,w_\alpha u x^i,L_n) \le 1/3 \}.
	\]
	If we only consider the prefixes of the form $w_\alpha u x^i$ for $i \in [1,m]$ we get the bound
	\[
		\frac{1}{8c} \ge \phi(R_n,w_\alpha u x^m,L_n,1/3) \ge \frac{m-|I(\alpha)|}{N} \ge \frac{m-|I(\alpha)|}{cm}
		\ge \frac{1}{c} - \frac{|I(\alpha)|}{cm}
	\]
	and therefore $|I(\alpha)| \ge \frac{7}{8}m$.
		
	Now we can use communication complexity and apply Lemma~\ref{lem:promise-idx} to prove the linear lower bound for $f(\R,n)$.
	Define
	\[
		D = \{ (\alpha,i) \colon \alpha \in \{0,1\}^m, i \in I(\alpha) \},
	\]
	which satisfies the condition of Lemma~\ref{lem:promise-idx}.
	Based on the algorithm $R_n$ we define a protocol for $\mathrm{IDX}_n|_D$.
	Consider the following procedure, where $\alpha = \alpha_1 \cdots \alpha_m \in \{0,1\}^m$ is the input for Alice,
	$i \in [1,m]$ is the input for Bob and $(\alpha,i)\in D$:
	Alice inputs the word $w_\alpha$ together with here random choice into $R_n$
	and sends an encoding of the reached state to Bob.
	Then Bob continues the simulation of $R_n$ from the received state and inputs the word $ux^i$
	and his random choice. He outputs $1$ if $R_n$ accepts $w_\alpha ux^i$, otherwise he outputs $0$.
	Since $(\alpha,i) \in D$, i.e., $i \in I(\alpha)$, we have $\eps(R_n,w_\alpha u x^i,L_n) \le 1/3$.
	Moreover, $w_\alpha u x^i \in L_n$ if and only if $\alpha_i = 1$.
	Hence, the probability that Bob outputs a wrong answer is bounded by $1/3$.
	By Lemma~\ref{lem:promise-idx} this implies that the above protocol has cost $\Theta(m) = \Theta(n)$, which in 
	turn implies that $f(\R,n) \notin o(n)$.
\end{proof}

\subsection{Deterministic lower bounds for arbitrarily small non-zero failure ratio}

In this section we prove the lower bound~\ref{lower-O(log n)-failure-det}
from Theorem~\ref{thm:trichotomy}. This is the only lower bound that is not
shown by a communication complexity argument. Instead, we use a simple 
combinatorial result on the ability of DFAs to count up to some threshold.

\subsubsection{Threshold counting}

In the following we show that, intuitively speaking, in order to count up to a threshold $n$
one needs $n$ states, even if a certain failure ratio $< 1/2$ is allowed.
A {\em counter with threshold $n$ and failure ratio $\phi < 1$}
is a deterministic streaming algorithm $A$ with
\[
	\phi(A,a^m, \{a^i \mid i \geq n \},0) \le \phi
\]
for all $m \ge n$. In other words: For every $m \ge n$ the number of $i \in [0,m]$ such that
$a^i \in L(A) \Leftrightarrow i \ge n$ does not hold is bounded by $\phi \cdot (m+1)$.

\begin{lemma}
	\label{lem:counter}
	Every counter with threshold $n$ and failure ratio $\phi < 1/2$ has at least $(1-2\phi) \cdot n$ many states.
\end{lemma}

\begin{proof}
	Take a counter $A = (Q,\{a\},q_0,\delta,F)$ with threshold $n$ and failure ratio $\phi < 1/2$, and
	consider the run
	\[
		q_0 \xrightarrow{a} q_1 \xrightarrow{a} q_2 \xrightarrow{a} \dots
	\]
	which is ultimately periodic.
	Let $t \ge 0$ and $d \ge 1$ be minimal such that $q_i = q_{i+d}$ for all $i \ge t$.
	Clearly we have $t+d \le |Q|$.
	Let $k_0 \ge 0$ be minimal such that $t+k_0d \ge n$.
	If $k_0 \in \{0,1\}$, then we are done because $|Q| \ge t+d \ge n$.
	Hence we can assume $k_0 \ge 2$ and by minimality $t+(k_0-1)d < n$.
	
	\medskip
	\noindent
	{\em Case 1.} Assume $|\{ q_t, \dots, q_{t+d-1} \} \setminus F| \ge d/2$.
	For all $k \ge k_0$ we have:
	\[
		\FT(A,a^{t+kd-1}, \{a^i \mid i \geq n \},0) \supseteq \bigcup_{j \in [k_0,k-1]} \{ i \in [t+jd,t+(j+1)d-1] : q_i \notin F \}
	\]
	and hence
	\[
		|\FT(A,a^{t+kd-1}, \{a^i \mid i \geq n \},0)| \ge (k-k_0) \cdot \frac{d}{2} \ge \frac{kd-(n+d-t)}{2} = \frac{t+kd}{2}-\frac{n+d}{2}.
	\]
	Dividing both sides by $t+kd$ yields
	\[
		\phi(A,a^{t+kd-1}, \{a^i \mid i \geq n \},0) \ge \frac{1}{2} - \frac{n+t}{2(t+kd)}.
	\]
	By taking $k$ large enough, we get $\phi(A,a^{t+kd-1}, \{a^i \mid i \geq n \},0) \ge \phi$, which is a contraction.
	
	\medskip
	\noindent
	{\em Case 2.} Assume $|\{ q_t, \dots, q_{t+d-1} \} \setminus F| < d/2$,
	or equivalently $|\{ q_t, \dots, q_{t+d-1} \} \cap F| > d/2$.
	We have
	\[
		\FT(A,a^{t+(k_0-1)d-1}, \{a^i \mid i \geq n \},0) \supseteq \bigcup_{k \in [1,k_0-1]} \{ i \in [t+(k-1)d,t+kd-1] : q_i \in F \}
	\]
	and hence
	\[
		|\FT(A,a^{t+(k_0-1)d-1}, \{a^i \mid i \geq n \},0)| > (k_0-1) \cdot \frac{d}{2} \ge \frac{n-t-d}{2}.
	\]
	Since $t+(k_0-1)d < n$, dividing both sides by $n$ yields: 
	\[
		\phi(A,a^{t+(k_0-1)d-1}, \{a^i \mid i \geq n \},0) > \frac{n-t-d}{2n}.
	\]
	Since $\phi(A,a^{t+(k_0-1)d-1}, \{a^i \mid i \geq n \},0) \le \phi$ it follows that
	$n-t-d < 2 \phi n$ and therefore $(1-2\phi) n < t+d \le |Q|$.
\end{proof}

\subsubsection{Construction of the witness strings}

In this section, we show the existence of certain witness strings 
that allow to construct from a DFA for a regular language, which does
not belongs to $\langle \LB, \PF, \SF, \Len \rangle$ 
a threshold counter with small failure ratio.
We start with a folklore fact on prefix-free languages:

\begin{lemma}
	\label{lem:pf-prop}
	Let $K,L$ be languages.
	\begin{enumerate}
		\item If $L$ is prefix-free and $K \subseteq L$, then $K$ is prefix-free as well. 
		\item If $K$ and $L$ are prefix-free, then $K  L$ is prefix-free as well.
	\end{enumerate}
\end{lemma}

\begin{proof}
	Point (1) is clear. For point (2) notice that the set of all prefixes of a fixed word
	is linearly ordered by the prefix relation.
	Consider $u,v \in K$, $x,y \in L$ and assume that $ux$ is a prefix of $vy$.
	Since either $u$ is a prefix of $v$, or vice versa,
	we must have $u = v$.
	Hence $ux$ is a prefix of $vy = uy$.
	Therefore $x$ is a prefix of $y$ and hence $x=y$.
	This proves $ux = vy$.
\end{proof}

\begin{lemma}
	\label{lem:concat-closure}
	Let $K \in \PF \cap \SF$ and $L \in \langle \PT, \PF, \Len \rangle$.
	Then $K L$ belongs to $\langle \RB, \PF, \Len \rangle$.
\end{lemma}

\begin{proof}
	Consider a representation of $L$ as a Boolean combination of languages in
	$\PT$, $\PF$, and $\Len$.
	We prove the statement by structural induction on this representation. Since every
	prefix-testable testable (resp., regular length language) is a Boolean combination
	of languages of the form $w \Sigma^*$ for some $w \in \Sigma^*$ (resp., 
	$\mathrm{MOD}_{q,r} = \Sigma^r (\Sigma^q)^*$ for some $q,r \ge 0$), it suffices
	to consider the three base cases in 1--3 below.
	
	\medskip
	\noindent
	{\em Case 1.} $L = w \Sigma^*$ for some $w \in \Sigma^*$.
	Since $\{w\} \in \PF \cap \SF$,  Lemma~\ref{lem:pf-prop} implies $K w \in \PF \cap \SF$,
	and hence $KL = K w \Sigma^* \in \RB$.
	
	\medskip
	\noindent
	{\em Case 2.} $L$ is prefix-free. Then $K L$ is also prefix-free by Lemma~\ref{lem:pf-prop}.
	
	\medskip
	\noindent
	{\em Case 3.} $L=\mathrm{MOD}_{q,r} = \Sigma^r (\Sigma^q)^*$ for some $q,r \ge 0$. 
	 If $q=0$ then $L = \Sigma^r$ is prefix-fee and we can go to Case~2.
	If $q > 0$, then we claim that
	\begin{equation}
		\label{eq:bifix-len}
		K L = \bigcup_{0 \le i < q} \big((K \cap \mathrm{MOD}_{q,i}) \Sigma^* \cap \mathrm{MOD}_{q,i+r}\big) .
	\end{equation}
	If $x \in K$ and $y \in L$, then let $i = |x| \bmod q$. We have $x \in K \cap \mathrm{MOD}_{q,i}$
	and $y \in \mathrm{MOD}_{q,r}$. Therefore $xy$ is contained on the right-hand side of \eqref{eq:bifix-len}.
	Conversely, if $w$ is contained on the right-hand side, then one can factorize $w = xy$ such that
	$x \in K \cap \mathrm{MOD}_{q,i}$. Since $|w| \equiv i+r \pmod{q}$ and $|x| \equiv i \pmod{q}$
	we must have $|y| \equiv r \pmod{q}$ which proves $y \in \mathrm{MOD}_{q,r}=L$.
	
	Now \eqref{eq:bifix-len} yields the desired Boolean combination: By Lemma~\ref{lem:pf-prop} $K \cap \mathrm{MOD}_{q,i}$
	is bifix-free and hence $(K \cap \mathrm{MOD}_{q,i})\Sigma^*$ belongs to $\RB$.
	Furthermore $\mathrm{MOD}_{q,i+r}$ belongs to $\Len$.
	
	\medskip
	\noindent
	{\em Case 4.} $L = L_1 \cup L_2$ with $L_1, L_2 \in \langle \PT, \PF, \Len \rangle$. Then
	\[
		K\, (L_1 \cup L_2) = K L_1 \cup K L_2
	\]
	and both $K L_1$ and $K L_2$ belong to $\langle \RB, \PF, \Len \rangle$ by induction.
	
	\medskip
	\noindent
	{\em Case 5.} 
	$L = \Sigma^* \setminus M$ and $M \in \langle \PT, \PF, \Len \rangle$. We claim that
	\[
		K \, (\Sigma^* \setminus M) = (\Sigma^* \setminus K M) \cap K \Sigma^*.
	\]
	If $w \in K \,(\Sigma^* \setminus M)$, then there exists a factorization $w=xy$ with $x \in K$.
	Since $K$ is prefix-free, this factorization is unique.
	Furthermore, we know $y \notin M$. Hence $w = xy \notin K M$ and $xy \in K \Sigma^*$.
	Conversely, if $w$ belongs to the right-hand side, then there exists a unique factorization $w=xy$ with $x \in K$.
	Since $xy \notin K M$, we know $y \notin M$ and therefore $xy \in K\,  (\Sigma^* \setminus M)$.
\end{proof}

\begin{theorem}
	\label{thm:id-sync-pair}
	Let $A = (Q,\Sigma,q_0,\delta,F)$ be a DFA for a language $L \subseteq \Sigma^*$.
	Suppose that every synchronized pair
	which is reachable from a positively idempotent state is $F$-consistent.
	Then $L$ belongs to $\langle \RB, \PF, \SF, \Len \rangle$.
\end{theorem}

\begin{proof}
	Let $Q_{\mathrm{id}} \subseteq Q$ be the set of states which are reachable from some positively idempotent state.
	First we decompose $L$ into
	\[
		L = L(A, F \cap Q_{\mathrm{id}})  \cup L(A, F \setminus Q_{\mathrm{id}}).
	\]
	By Lemma~\ref{lem:before-idempotent} $L(A, F \setminus Q_{\mathrm{id}})$ belongs to $\langle \SF \rangle$.
	If $q_0 \in Q_{\mathrm{id}}$, then by the assumption in the lemma every synchronized pair
	which is reachable from $q_0$ is $F$-consistent. By Lemma~\ref{lem:sync-pair},
	$L$ belongs to $\langle \PT, \PF, \Len \rangle$.
	Therefore $L$ is also contained in $\langle \RB, \PF, \Len \rangle$
	by Lemma~\ref{lem:concat-closure} (take $K=\{\varepsilon\}$).
	Now let us assume that $q_0 \notin Q_{\mathrm{id}}$, i.e., $q_0$ is not reachable from a positively idempotent state
	and thus $\eps \notin L(A, F \cap Q_{\mathrm{id}})$.
	The idea is to factorize a word $w \in L(A, F \cap Q_{\mathrm{id}})$ into the form $w = uav$
	where $u,v \in \Sigma^*$, $a \in \Sigma$ and $ua$ is the minimal prefix such that $\delta(q_0,ua) \in Q_{\mathrm{id}}$.
	To do so define the set $\Delta$ of transitions
	which lead from $Q \setminus Q_{\mathrm{id}}$ to $Q_{\mathrm{id}}$:
	\[
		\Delta = \{ (p,a,q) : p \in Q \setminus Q_{\mathrm{id}}, \, q \in Q_{\mathrm{id}}, \, a \in \Sigma, \, \delta(p,a) = q  \}
	\]
	Consider a triple $(p,a,q) \in \Delta$.
	Define the language
	\[
		L_q = \{ x \in \Sigma^* : \delta(q,x) \in F \},
	\]
	which belongs to $\langle \PT, \PF, \Len \rangle$ by Lemma~\ref{lem:sync-pair}.
	Further, we define the DFA $B_p$ which is obtained from $A$
	by defining $p$ to be the only final state and making all states from $Q_{\mathrm{id}}$ looping.
	Note that set of the positively idempotent states in $B_p$ is exactly $Q_{\mathrm{id}}$.
	We then have
	\[
		L(A, F \cap Q_{\mathrm{id}}) = \bigcup_{(p,a,q) \in \Delta} L(B_p) \, a \, L_q.
	\]
	Consider a triple $(p,a,q) \in \Delta$.
	Again by Lemma~\ref{lem:before-idempotent}, the language $L(B_p)$ is a finite union of regular suffix-free languages, say
	\[
		L(B_p) = \bigcup_{i=1}^k L_i,
	\]
	and we have
	\[
		L(B_p) \, a \, L_q = \bigcup_{i=1}^k (L_i \, a \, L_q).
	\]
	Since each $L_i$ and $\{a\}$ are suffix-free, $L_i  a$ is also suffix-free.
	
	We claim that  $L(B_p) a$ is also prefix-free. 
	Assume that there exist $x,y \in L(B_p)$ such that $xa$ is a proper prefix of $ya$.
	Hence $xa$ is a prefix of $y$. In $B_p$, $x$ leads from the initial state to the final state $p$. 
	Since $\delta(p,a) = q$ and $p \not\in Q_{\mathrm{id}}$, $xa$ leads in $B_p$ to $q$.
	Since $q \in Q_{\mathrm{id}}$ is looping in $B_p$, also $y$ leads to $q$. But this
	contradicts the fact that $y$ is accepted by $B_p$. 
	
	Since $L(B_p)a$ is prefix-free, also every subset $L_i a$ is prefix-free and hence
	bifix-free. By Lemma~\ref{lem:concat-closure} this implies that $L_i  a  L_q$ belongs to
	$\langle \RB, \PF, \Len \rangle$.
	This concludes the proof.
\end{proof}

	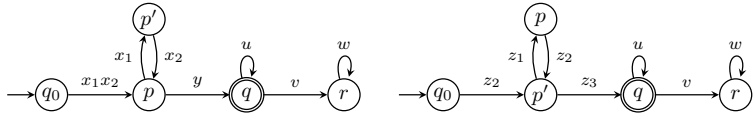
\begin{figure}
			\centering
			
			\scalebox{0.8}{
			\begin{tikzpicture}[semithick,->,>=stealth]
			\node[state, initial left, initial text={}] (q0) {$q_0$};
			\node[state, right = 30pt of q0] (p) {$p$};
			\node[state, above = 20pt of p] (p') {$p'$};
			\node[state, double, right = 30pt of p] (q) {$q$};
			\node[state, right = 30pt of q] (r) {$r$};
			
			\draw (q0) edge node[above] {\footnotesize $x_1x_2$} (p);
			\draw (p) edge[bend left = 15] node[left] {\footnotesize $x_1$} (p');
			\draw (p') edge[bend left = 15] node[right] {\footnotesize $x_2$} (p);
			\draw (p) edge node[above] {\footnotesize $y$} (q);
			\draw [loop above] (q) edge node[above] {\footnotesize $u$} (q);
			\draw (q) edge node[above] {\footnotesize $v$} (r);
			\draw [loop above] (r) edge node[above] {\footnotesize $w$} (r);
			\end{tikzpicture}
			}
			\scalebox{0.8}{
			\begin{tikzpicture}[semithick,->,>=stealth]
			\node[state, initial left, initial text={}] (q0) {$q_0$};
			\node[state, right = 30pt of q0] (p) {$p'$};
			\node[state, above = 20pt of p] (p') {$p$};
			\node[state, double, right = 30pt of p] (q) {$q$};
			\node[state, right = 30pt of q] (r) {$r$};
			
			\draw (q0) edge node[above] {\footnotesize $z_2$} (p);
			\draw (p) edge[bend left = 15] node[left] {\footnotesize $z_1$} (p');
			\draw (p') edge[bend left = 15] node[right] {\footnotesize $z_2$} (p);
			\draw (p) edge node[above] {\footnotesize $z_3$} (q);
			\draw [loop above] (q) edge node[above] {\footnotesize $u$} (q);
			\draw (q) edge node[above] {\footnotesize $v$} (r);
			\draw [loop above] (r) edge node[above] {\footnotesize $w$} (r);
			\end{tikzpicture}
			}
		\caption{Case 1. from the proof of Lemma~\ref{lem:rb-pf-sf-len}
		(there is a symmetric case, where $p' \in F$ and $q \notin F$).}
		\label{fig:rb-pf-sf-len-proof}
	\end{figure}

\begin{lemma}
	\label{lem:rb-pf-sf-len}
	If $L \in \Reg \setminus \langle \RB, \PF, \SF, \Len \rangle$,
	then there exist words $u,v,w,x \in \Sigma^*$ with
	$|uv| = |w| = |x| \ge 1$ and
	\begin{itemize}
		\item $v(uv)^* \subseteq L$ and $v(uv)^*wx^* \cap L = \emptyset$, or
		\item $v(uv)^* \cap L = \emptyset$ and $v(uv)^*wx^* \subseteq L$
	\end{itemize}
\end{lemma}

\begin{proof}
	Let $A = (Q,\Sigma,q_0,\delta,F)$ be a DFA for $L$.
	By Theorem~\ref{thm:id-sync-pair} there exists a positively idempotent state $p \in Q$
	and a synchronized state pair $(q,r)$ reachable from $p$ which is not $F$-consistent.
	This means that
	\begin{itemize}
		\item $|\{q,r\} \cap F| = 1$,
		\item there exists a word $x \in \Sigma^+$ with $\delta(q_0,x) = p$ and $\delta(p,x) = p$,
		\item there exists a word $y \in \Sigma^*$ with $\delta(p,y) = q$,
		\item there exist words $u,v,w \in \Sigma^+$ of the same length such that
		$\delta(q,u) = q$, $\delta(q,v) = r$ and $\delta(r,w) = r$.
	\end{itemize}
	We emphasize that these are not the words from the lemma.
	We can ensure that $|x| \ge |y|$ by replacing $x$ by $x^{|y|}$.
	Furthermore, we can ensure that $|x|=|u|=|v|=|w|$ by replacing $x$ by $x^{|u|}$,
	$u$ by $u^{|x|}$, $v$ by $u^{|x|-1}v$ and $w$ by $w^{|x|}$.

	Let $x = x_1 x_2$ such that $|x_1| = |y|$, or equivalently $|x| = |x_2y|$.
	Define the state $p' = \delta(p,x_1)$ and the words $z_1 = x_2$, $z_2 = x_1x_2x_1$ and $z_3 = x_2 y$.
	The situation is depicted in Figure~\ref{fig:rb-pf-sf-len-proof}.
	Since $(q,r)$ is not $F$-consistent, either $(p',q)$ or $(p',r)$ is not $F$-consistent.
	
	\medskip
	\noindent
	{\em Case 1.} Assume that $(p',q)$ is not $F$-consistent.
	Then we can take the words $z_1, z_2, z_3u, uu$ because
	$|z_1z_2| = |z_3u| = |uu|$ and for all $i \ge 0$ we have:
	\begin{itemize}
		\item $\delta(q_0,z_2(z_1z_2)^i) = p'$ and
		\item $\delta(p',z_3u(uu)^i) = q$.
	\end{itemize}

	\medskip
	\noindent	
	{\em Case 2.} Assume that $(p',r)$ is not $F$-consistent.
	Then we can take the words $z_1, z_2, z_3v, ww$ because
	$|z_1z_2| = |z_3v| = |ww|$ and for all $i \ge 0$ we have:
	\begin{itemize}
		\item $\delta(q_0,z_2(z_1z_2)^i) = p'$ and
		\item $\delta(p',z_3v(ww)^i) = r$.
	\end{itemize}
	This concludes the proof of the lemma.
	The final forbidden pattern is shown in Figure~\ref{fig:rb-pf-sf-len}.
\end{proof}

	\begin{figure}
			\centering
			\begin{tikzpicture}[semithick,->,>=stealth]
			\node[state, initial left, initial text={}] (q0) {};
			\node[state, double, right = 30pt of q0] (p) {\footnotesize $q_1$};
			\node[state, above = 20pt of p] (r) {};
			\node[state, right = 30pt of p] (q) {\footnotesize $q_2$};
			
			\draw (q0) edge node[above] {\footnotesize $v$} (p);
			\draw (p) edge[bend left = 15] node[left] {\footnotesize $u$} (r);
			\draw (r) edge[bend left = 15] node[right] {\footnotesize $v$} (p);
			\draw (p) edge node[above] {\footnotesize $w$} (q);
			\draw [loop above] (q) edge node[above] {\footnotesize $x$} (q);
			\end{tikzpicture}
		\caption{Forbidden pattern for $\langle \RB, \PF, \SF, \Len \rangle$ where $|uv| = |w| = |x| \ge 1$
		(there is a symmetric case, where $q_1 \notin F$ and $q_2 \in F$).}
		\label{fig:rb-pf-sf-len}
	\end{figure}
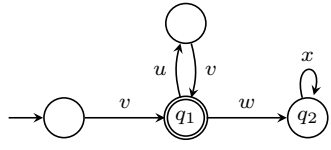

Lemma~\ref{lem:rb-pf-sf-len} applied to 
the language $L^{\rev}$ yields:

\begin{lemma}
	\label{lem:lb-pf-sf-len}
	If $L \in \Reg \setminus \langle \LB, \PF, \SF, \Len \rangle$,
	then there exist words $u,v,w,x \in \Sigma^*$ with
	$|uv| = |w| = |x| \ge 1$ and
	\begin{itemize}
		\item $(vu)^*v \subseteq L$ and $x^*w(vu)^*v \cap L = \emptyset$, or
		\item $(vu)^*v \cap L = \emptyset$ and $x^*w(vu)^*v\subseteq L$.
	\end{itemize}
\end{lemma}
Using the witness strings from Lemma~\ref{lem:lb-pf-sf-len}
we can now prove point~\ref{lower-O(log n)-failure-det}
from Theorem~\ref{thm:trichotomy}.

\begin{theorem}
	If $L \in \Reg \setminus \langle \LB, \PF, \SF, \Len \rangle$
	then there exists a failure ratio $0 < \phi < 1$
	such that every \SWA\ $\R$ for $L$ with failure ratio $\phi$
	satisfies $f(\R,n) \notin o(\log n)$.
\end{theorem}

\begin{proof}
	Let $u,v,w,x$ be the words from Lemma~\ref{lem:lb-pf-sf-len}.
	Without loss of generality we assume the first case from the lemma,
	i.e., $(vu)^*v \subseteq L$ and $x^*w(vu)^*v \cap L = \emptyset$.
	Let $\R = (R_n)_{n \ge 0}$ be a \SWA\ for $L$ with failure ratio $\phi$, which is chosen later.
	Let $m \ge 0$ be a natural number and $n = m \cdot |x| + |v| = \Theta(m)$.
	Define $y_{m,i} = x^{m-1}wv(uv)^i$, which has length $(m+i) \cdot |x| + |v|$.
	Observe that
	\[
		 \last_{n}(x^{m-1}wv(uv)^i) = \begin{cases}
			x^{m-1-i}wv(uv)^i, & \text{if } i < m, \\
			v(uv)^m, & \text{if } i \ge m,
		\end{cases}
	\]
	and thus
	\[
		\last_n(y_{m,i}) \in L \iff i \ge m.
	\]
	Now consider the streaming algorithm $R_n = (Q,\Sigma,q_0,\delta,F)$ for window length $n$.
	It suffices to show that $|Q| \in \Omega(m)$.
	Consider the (infinite) run of $R_n$ on $x^{m-1}wv(uv)(uv)(uv)\cdots$:
	\[
		q_0 \xrightarrow{x^{m-1}wv} p_0 \xrightarrow{uv} p_1 \xrightarrow{uv} p_2 \xrightarrow{uv} \cdots
	\]
	Define the DFA $C = (Q,\{a\}, p_0, \mu, F)$ over $\{a\}$ where
	$\mu(q,a) = \delta(q,uv)$ for all $q \in Q$,
	which is a counter with threshold $m$ and a certain failure ratio.
	For $k \ge m$ let $e_k$
	be the number of prefixes of the form $y_{m,i}$ of $y_{m,k}$ on which $R_n$ errs.
	Then the failure ratio of $C$ is bounded by $\sup_{k \ge m} e_k/(k+1)$.
	For all $k \ge m$ we have:
	\[
		\phi(R_n,y_{m,k-1},L_n,0) \ge \frac{e_k}{(m+k)|x|+|v|+1}
	\]
	Since $|v| \le |x| > 0$, $m \le k$, and
	$\phi(R_n,y_{m,k},L_n,0) \le \phi$ it follows that
	\[
	\phi \ge \frac{e_k}{(m+k+1)|x|+1} \ge \frac{e_k}{(2k+1)|x|+1} \ge \frac{e_k}{2(k+1)|x|}
	\]
	and therefore $e_k/(k+1) \le 2 \phi |x|$.
	By choosing $\phi < 1/(4|x|)$, we obtain a counter $C$
	with threshold $m$ and failure ratio $< \frac{1}{2}$.
	By Lemma~\ref{lem:counter} we know that $C$ has $\Omega(m)$ many states,
	which concludes the proof.
\end{proof}

\section{One-sided error} \label{sec-one-sided}

So far, we have only considered randomized \SWAs\ with a two sided error (analogously to the 
complexity class {\sf BPP}). Randomized \SWAs\ with a one-sided error (analogously to the 
class {\sf RP})  as defined below
can be motivated by applications, where all ``yes'' outputs have to be correct. 
Formally, a randomized \SWA\ $\R = (R_n)_{n \geq 0}$
has {\em one sided error} for $L \subseteq \Sigma^*$ if the following holds
for all $n \geq 0$ and words $w \in \Sigma^*$:
\begin{itemize}
\item If $w \notin L_n$ then $\eps(R_n,w,L_n) = 0$.
\item If $w \in L_n$ then $\eps(R_n,w,L_n) \leq 1/2$.
\end{itemize}
In other words: If $w \notin L_n$ then $R_n$ rejects $w$ with probability $1$ and
if $w \in L$ then $R_n$ accepts $w$ with probability at least $1/2$. The choice
of $1/2$ is arbitrary. Using a standard probability amplification argument, we can 
replace $1/2$ by any non-zero probability.

\begin{theorem} \label{thm:trichotomy-one-sided}
	Let $L \subseteq \Sigma^*$ be a regular language.
	\begin{enumerate}
	\item  If $L \in \langle \ST, \Len \rangle$,
	then $L$ has a randomized \SWA\ with one-sided error and $f(\R,n) = \mathcal{O}(1)$.
	\item If $L \notin \langle \ST, \Len \rangle$, then $f(\R,n) \notin o(\log n)$ for every randomized \SWA\ for $L$ with one-sided error.
	\item If $L \in \langle \LI, \Len \rangle$,
	then $L$ has a randomized \SWA\ with one-sided error and $f(\R,n) = \mathcal{O}(\log n)$.
	\item If $L \notin \langle \LI, \Len \rangle$, then $f(\R,n) \notin o(n)$ for every randomized \SWA\ for $L$ with one-sided error.
	\end{enumerate}
\end{theorem}
 We only have to prove point (2) of the theorem: The upper bounds in (1) and (3) already hold for deterministic \SWAs\ 
 \cite{GHKLM18,GHL16}.
 Moreover, the lower bound in (4) follows from point \eqref{point-O(log)} in Theorem~\ref{thm:quatrochotomy}.
 
 In order to show point (2) from Theorem~\ref{thm:trichotomy-one-sided} we prove a stronger statement.
 Note that if $\R = (R_n)_{n \geq 0}$ is a randomized \SWA\ for $L$ with one-sided error, then every $R_n$
 can be seen as an NFA (non-deterministic finite automaton) for $L_n$ by forgetting probabilities. Hence, it suffices to show:
 
 \begin{proposition} \label{prop-lower-sqrt}
Let $L \in \Reg \setminus \langle \ST, \Len \rangle$, $n \geq 0$,  and let $A$ be an NFA for $L_n$. 
Then, $A$ has $\Omega(\sqrt{n})$ many states. 
\end{proposition}

Let us first fix the notation concerning NFAs. An NFA is a tuple $A = (Q,\Sigma, I, \Delta, F)$, where
$Q$ is the finite set of state, $\Sigma$ is the input alphabet, $I \subseteq Q$ is the set of initial states,
$\Delta \subseteq Q \times \Sigma \times Q$ is the set of transitions, and $F \subseteq Q$ is the set of 
final states.  We define $\Delta^* \subseteq Q \times \Sigma^* \times Q$ as the smallest relation such that
(i) $(q, \varepsilon, q) \in \Delta^*$ for all states $q \in Q$ and (ii) $(p,w,q) \in \Delta^*$ and $(q,a,r) \in \Delta$
implies $(p,wa,r) \in \Delta^*$. The language accepted by $A$ is $L(A) = \{ w \in \Sigma^* \colon \exists p \in I, q \in F 
\colon (p,w,q) \in \Delta^* \}$.
For the proof of Proposition~\ref{prop-lower-sqrt} we need the following lemma.

\begin{lemma} \label{lemma-NFA-lower}
Let $L \subseteq a^*$ such that there exists an $n$ with $a^n \notin L$ and $a^k \in L$ for all $k > n$.
Then, every NFA for $L$ has at least $\sqrt{n}$ many states.
\end{lemma}

\begin{proof}
The proof is essentially the same as for \cite[Lemma~6]{JiraskovaM14}, where the statement of the lemma is shown for $L = a^* \setminus \{a^n\}$.
Let us give the proof for completeness. 
It is known that every unary NFA has an equivalent NFA in so called Chrobak normal form.  A unary NFA in Chrobak normal
form consists of path starting in the unique initial state. From  the last state of the path, edges go to a collection of 
disjoint cycles. In \cite{Gawrychowski11} it is shown that an $m$-state unary NFA has an equivalent NFA in Chrobak normal form whose
initial path consists of $m^2 - m$ states. Now assume that $L$ is accepted by an NFA with $m$ states and let $A$ 
be the equivalent Chrobak normal form NFA, whose initial path consists of $m^2-m$ states. If $n \geq m^2-m$ then all states
that are reached in $A$ from the initial state via $a^n$ belong to a cycle and every cycle contains such a state. Since
$a^n \notin L$, all these states are rejecting. Hence, $a^{n+ x \cdot d} \notin L$ for all $x \geq 0$, where $d$ is the product
of all cycle lengths. This contradicts the fact that $a^k \in L$ for all $k > n$. Hence, we must have $n < m^2-m$ and therefore
$m > \sqrt{n}$.
\end{proof}

\begin{proof}[Proof of Proposition~\ref{prop-lower-sqrt}]
Since $L \in \Reg \setminus \langle \ST, \Len \rangle$, we can apply Lemma~\ref{lem:pt-len} to the language $L^\rev$
and obtain words $u,x,y$ such that $|x|=|y| \geq 1$ and on the two following cases holds:
\begin{enumerate}[(i)]
\item $x^*u \cap L = \emptyset$ and $y x^* u \subseteq L$
\item $x^*u \subseteq L$ and $y x^* u \cap L = \emptyset$
\end{enumerate}
Note that we must have $x \neq y$.

Fix an $m \geq 0$ and consider the window size $n  =  (m+1) |x| + |u|$.
Let us first assume that (i) holds. 
Consider the words $x_i = y x^i$ and $y_i = x^{m-i} u$ for $0 \leq i \leq m$.
We have $x_i y_i = y x^m u  \in L_n$ and $x_i y_j = y x^{m-i+j} u \notin L_n$ for all
$i < j$. The fooling set technique from \cite[Lemma~1]{Birget92} implies that every NFA for $L_n$
has at least $m+1 \in \Omega(n)$ many states.

Now assume that (ii) holds. Assume that $A = (Q,\Sigma,I,\Delta,F)$ is an NFA for $L_n$.
We define an NFA $A'$ over the unary alphabet $\{a\}$ as follows:
\begin{itemize}
\item The state set of $A'$ is $Q$.
\item The set of initial states of $A'$ is $\{ q \in Q \colon \exists p \in I \colon (p,y,q) \in \Delta^* \}$.
\item The set of final states of $A'$  is $\{ p \in Q \colon \exists q \in F \colon (p,u,q) \in \Delta^* \}$.
\item The set of transitions of $A'$ is $\{ (p,a,q) \colon (p,x,q) \in \Delta^* \}$.
\end{itemize}
We then have the following two properties:
\begin{itemize}
\item If $k > m$, then $y x^k u \in L_n = L(A)$, which implies $a^k \in L(A')$.
\item $y x^m u \notin L(A)$, which implies $a^m \notin L(A')$.
\end{itemize}
By Lemma~\ref{lemma-NFA-lower}, $A'$ (and thus $A$) has at least $\sqrt{m} \in \Omega(\sqrt{n})$ many
states.
\end{proof}

\section{Strict error probability} \label{sec-strict}

Let $\Pi \subseteq \Sigma^* \times \Omega$ be an approximation problem and let $\R = (R_n)_{n \ge 0}$ be a randomized \SWA\ which is strictly $\eps$-correct for $\Pi$, where $0 \leq \eps < 1$.
In this section we will prove that one can extract a deterministic \SWA\ $\D = (D_n)_{n \ge 0}$
for $\Pi$ from $\R$ such that $f(\D,n) \le f(\R,n)$ for all $n \geq 0$. Since we deal with the worst case 
space complexity of $\R$, we can assume that every $R_n$ has a finite state set. 

Fix a window size $n \ge 0$ and let $R_n = (Q,\Sigma,\iota,\rho,\omega)$.
Consider a run 
\[ \pi : q_0 \xrightarrow{a_1} q_1 \xrightarrow{a_2} \cdots \xrightarrow{a_m} q_m
\]
in $R_n$.
A {\em subrun} of $\pi$ is a run of the form 
\[ q_i \xrightarrow{a_{i+1}} q_{i+1} \xrightarrow{a_{i+2}} \cdots q_{j-1} \xrightarrow{a_j} q_j .
\]
The run $\pi$ is {\em simple} if $q_i \neq q_j$ for $0 \leq i < j \leq m$.
Consider a nonempty subset $S \subseteq Q$ and a function $\delta \colon Q \times \Sigma \to Q$
such that $S$ is closed under $\delta$, i.e., $\delta(S \times \Sigma) \subseteq S$.
We say that the run $\pi$ is {\em $\delta$-conform}
if $\delta(q_{i-1},a_i) = q_i$ for all $1 \le i \le m$.
We say that $\pi$ is {\em $(S,\delta)$-universal} if for all $q \in S$ and $x \in \Sigma^n$
there exists a $\delta$-conform subrun $\pi' : q \xrightarrow{x} q'$ of $\pi$. 
Finally, $\pi$ is {\em $\delta$-universal} if it is $(S,\delta)$-universal
for some nonempty subset $S \subseteq Q$ which is closed under $\delta$. 

\begin{lemma}
	\label{lem:extract-det}
	Let $\pi$ be a strictly correct run in $R_n$ for $\Pi$,
	let $S \subseteq Q$ be a nonempty subset and let $\delta \colon Q \times \Sigma \to Q$ be a function
	such that $S$ is closed under $\delta$.
	If $\pi$ is $(S,\delta)$-universal, then there exists $q_0 \in S$ such that
	$D_n = (Q,\Sigma,q_0,\delta,\omega)$ is a deterministic streaming algorithm for $\Pi_n$.
\end{lemma}

\begin{proof}
	Let $q_0 = \delta(p,\square^n) \in S$ for some arbitrary state $p \in S$
	and define $D_n = (Q,\Sigma,q_0,\delta,\omega)$.
	Let $w \in \Sigma^*$ and consider the run $\sigma: p \xrightarrow{\square^n} q_0 \xrightarrow{w} q$ in $D_n$
	of length $\ge n$. We have to show that $(\last_n(w),\omega(q)) \in \Pi$. We can write
	$\square^n w = x \, \last_n(w)$ for some $x \in \Sigma^*$. Thus, we can rewrite the run $\sigma$ as
	$\sigma : p \xrightarrow{x} q' \xrightarrow{\last_n(w)} q$.
	We know that $q' \in S$ because $S$ is closed under $\delta$.
	Since $\pi$ is $(S,\delta)$-universal, it contains a subrun $q' \xrightarrow{\last_n(w)} q$.
	By strict correctness of $\pi$ we obtain $(\last_n(w),\omega(q)) \in \Pi$.
\end{proof}

For the rest of this section we fix an arbitrary function $\delta \colon Q \times \Sigma \to Q$
such that for all $q \in Q$, $a \in \Sigma$,
\[
\rho(q,a,\delta(q,a))  = \max \{ \rho(q,a,p) \colon p \in Q \} .
\]
Note that 
\[
\rho(q,a,\delta(q,a)) \geq \frac{1}{|Q|} .
\]
 for all $q \in Q$, $a \in \Sigma$.
Furthermore, let $D_n = (Q,\Sigma,q_0,\delta,\omega)$ where the initial state $q_0$ will be defined later.
We define for each $i \ge 1$ a state $p_i$, a run $\pi^*_i \in \Runs(D_n,p_i,w_i)$ in $D_n$
and a set $S_i \subseteq Q$.
We abbreviate $\Runs(R_n,w_1 \cdots w_m)$ by $R_m$.
For $1 \leq i \leq m$ let $H_i$ denote the event that for a random run
$\pi = \pi_1 \cdots \pi_m \in R_m$, where each $\pi_j$ is a run on $w_j$,
the subrun $\pi_i$ is $(S_i,\delta)$-universal.
Notice that $H_i$ is independent of $m \ge i$.

First, we choose for $p_1$ a state that maximizes
\[
	\Pr_{\pi \in R_{i-1}}[\pi \text{ ends in } p_i \mid \forall j \leq i-1 : \overline{H_j}  ],
\]
which is at least $1/|Q|$. Note that $p_1$ is a state such that $\iota(p_1)$ is maximal, since $R_0$ only
consists of empty runs $(q)$. For $S_i$ we take any maximal SCC of $D_n$ which is reachable from $p_i$.
Finally, we define the run $\pi^*_i$. It  starts in $p_i$.
Then, for each state $q \in S_i$ and each word $x \in \Sigma^n$ the run $\pi^*_i$ leads from the current state to $q$
via a simple run and reads the word $x$ from $q$. Since $S_i$ is a maximal SCC of $D_n$ such a run exists.
Hence, $\pi^*_i$ is a run on a word of the form
\[
	w_i = \prod_{q \in S_i} \prod_{x \in \Sigma^n} y_{q,x} \, x.
\]
Since we choose the runs on the words $y_{q,x}$ to be simple,
the lengths of the words $w_i$ are bounded independently of $i$.
More precisely, we have $|w_i| \leq |Q| \cdot |\Sigma|^n \cdot (|Q|+n)$.
Let us define
\[
\mu = \frac{1}{|Q|^{|Q| \cdot |\Sigma|^n \cdot (|Q|+n) + 1}} .
\]

\begin{lemma}
	\label{lem:mu}
	For all $m \ge 0$ we have
	\[
		\Pr_{\pi \in R_m}[H_m \mid \forall i\leq m-1: \overline{H_i}] \geq \mu .
	\]
\end{lemma}

\begin{proof}
         In the following, let $\pi$ be a random run from $R_m$ and let $\pi_i$ be the subrun on $w_i$.
	Notice that under the assumption that the event $[\pi_{m-1}$ ends in $p_m]$ holds, the events $[\pi_m = \pi^*_m]$ and 
	$[\forall i \leq m-1: \overline{H_i}]$
	are conditionally independent.\footnote{Two events $A$ and $B$ are conditionally independent assuming event $C$ if 
	$\Pr[A \wedge B \mid C] = \Pr[A \mid C] \cdot \Pr[B \mid C]$, which is equivalent to 
	$\Pr[A  \mid B \wedge C] = \Pr[A  \mid  C]$.} Thus, we have
	\begin{eqnarray*} && \Pr_{\pi \in R_m}[\pi_m = \pi^*_m  \mid \pi_{m-1} \text{ ends in } p_m \wedge \forall i \leq m-1: \overline{H_i}] \\
	&=& \Pr_{\pi \in R_m}[\pi_m = \pi^*_m  \mid \pi_{m-1} \text{ ends in } p_m] .
	\end{eqnarray*}
	 Since the event $[\pi_m = \pi^*_m]$ implies the event 
	 $[\pi_{m-1}$ ends in $p_m]$, 
	 we obtain:
	\begin{eqnarray*}
		&& \Pr_{\pi \in R_m}[H_m \mid \forall i \leq m-1: \overline{H_i}] \\
		&\ge& \Pr_{\pi \in R_m}[\pi_m = \pi_m^* \mid \forall i \leq m-1: \overline{H_i}] \\
		&=& \Pr_{\pi \in R_m}[\pi_m = \pi_m^* \wedge \pi_{m-1} \text{ ends in } p_m \mid \forall i \leq m-1: \overline{H_i}] \\
		&=& \Pr_{\pi \in R_m}[\pi_m = \pi_m^*  \mid  \pi_{m-1} \text{ ends in } p_m \wedge \forall i \leq m-1: \overline{H_i}] \cdot \\
		 &&   \Pr_{\pi \in R_m}[\pi_{m-1} \text{ ends in } p_m \mid \forall i \leq m-1: \overline{H_i}]  \\
		&=& \Pr_{\pi \in R_m}[\pi_m = \pi^*_m \mid \pi_{m-1} \text{ ends in } p_m]
		\cdot \\
		&& \Pr_{\pi \in R_m}[\pi_{m-1} \text{ ends in } p_m \mid \forall i \leq m-1: \overline{H_i}] \\
		&\ge & \Pr_{\pi_m \in \mathrm{Runs}(p_m,w_m)} [\pi_m = \pi^*_m] \cdot \frac{1}{|Q|}\\
		& \ge & \frac{1}{|Q|^{|w_m|+1}} \geq \mu
	\end{eqnarray*}
	This proves the lemma.
\end{proof}

\begin{lemma}
	\label{lem:almost-surely-universal}
	$\Pr_{\pi \in R_m}[\pi \text{ is $\delta$-universal}] \ge \Pr_{\pi \in R_m}[\exists i \leq m: H_i] \ge 1 - (1-\mu)^m$. 
\end{lemma}

\begin{proof}
	The first inequality follows from the definition of the event $H_i$. Moreover, 
	we have
	\begin{eqnarray*}
		\Pr_{\pi \in R_m}[\exists i \leq m : H_i]
		&=& \Pr_{\pi \in R_m}[\exists i \leq m-1: H_i]
		+  \\
		&&\Pr_{\pi \in R_m}[H_m \mid \forall i \leq m-1: \overline{H_i}] \cdot
		\Pr_{\pi \in R_m}[\forall i \leq m-1: \overline{H_i}] \\
		&=&
		\Pr_{\pi \in R_{m-1}}[\exists i \leq m-1: H_i]
		+ \\
		&&\Pr_{\pi \in R_m}[H_m \mid \forall i \leq m-1: \overline{H_i}] \cdot
		\Pr_{\pi \in R_{m-1}}[\forall i \leq m-1: \overline{H_i}] \\
		&\ge&
		\Pr_{\pi \in R_{m-1}}[\exists i \leq m-1: H_i]
		+ \mu \cdot
		\Pr_{\pi \in R_{m-1}}[\forall i \leq m-1: \overline{H_i}].
	\end{eqnarray*}
	Define $r_m = \Pr_{\pi \in R_m}[\exists i \leq m : H_i]$. We get
	\begin{equation*} 
           r_m \ge r_{m-1} + \mu \cdot (1-r_{m-1}) = (1-\mu) \cdot r_{m-1} + \mu.
         \end{equation*}
         Since $r_0 = 0$, we get $r_m \ge 1 - (1-\mu)^m$ by induction.
\end{proof}

\begin{theorem} \label{thm-strict}
	There exists $q_0 \in Q$ such that $D_n = (Q,\Sigma,q_0,\delta,\omega)$
	is a deterministic streaming algorithm for $\Pi_n$.
\end{theorem}

\begin{proof}
	We use the probabilistic method. With Lemma~\ref{lem:almost-surely-universal} we get
	\begin{eqnarray*}
		&& \Pr_{\pi \in R_m}[\pi \text{ is strictly correct for $\Pi$ and $\delta$-universal}] \\
		& = & 1 - \Pr_{\pi \in R_m}[\pi \text{ is not strictly correct for $\Pi$ or is not $\delta$-universal}] \\
		& \ge & 1 - \Pr_{\pi \in R_m}[\pi \text{ is not strictly correct for $\Pi$}] - \Pr_{\pi \in R_m}[\pi \text{ is not $\delta$-universal}] \\
		&\ge & \Pr_{\pi \in R_m}[\pi \text{ is $\delta$-universal}] - \eps \\
		&\ge & 1 - (1-\mu)^m - \eps .
	\end{eqnarray*}
	We have $1 - (1-\mu)^m - \eps > 0$ for $m > \log(1-\eps) / \log(1-\mu)$ (note that $\eps < 1$ and $0 < \mu < 1$). 
	Hence there exists an $m \ge 0$ and a strictly correct run $\pi \in R_m$ which is $\delta$-universal.
	The statement follows directly from Lemma~\ref{lem:extract-det}.
\end{proof}

\begin{corollary}
	There exists a deterministic sliding window algorithm $\D$
	for $\Pi$ such that $f(\D,n) \le f(\R,n)$ for all $n \geq 0$.
\end{corollary}

The word $w_1 w_2 \cdots w_m$ (with $m > \log(1-\eps) / \log(1-\mu)$), for which there exists
a strictly correct and $\delta$-universal run  has a length that is exponential in the window size $n$.
In other words: We need words of length exponential in $n$ in order to transform a strictly $\eps$-correct
randomized \SWA\ into an equivalent deterministic \SWA. We remark that this is unavoidable: if 
we restrict to inputs of length $\mathrm{poly}(n)$ then strictly $\eps$-correct \SWAs\ can yield a proper 
space improvement over deterministic \SWAs.

Take the language $K_{\mathrm{pal}} = \{ ww^\rev : w \in \{a,b\}^n \}$ of all palindromes of even length,
which belongs to the class $\mathbf{DLIN}$ of deterministic linear context-free languages,
and let $L = \$ K_{\mathrm{pal}}$.

\begin{lemma}
	If $\D$ is a deterministic SWA for $L$, then $f(\D,2n+1) = \Omega(n)$.
\end{lemma}

\begin{proof}
	Take two distinct words $\$ x$ and $\$ y$ where $x,y \in \{a,b\}^n$.
	Since $D_{2n+1}$ accepts $\$ x x^\rev$ and rejects $\$ y x^\rev$,
	the automaton $D_n$ reaches two different states on the inputs $\$ x$ and $\$ y$.
	Therefore, $D_{2n+1}$ must have at least $|\{a,b\}^n| = 2^n$ states.
\end{proof}
Let us now fix a polynomial $p(n)$.

\begin{lemma}
There is a randomized \SWA\ $\R = (R_n)_{n \geq 0}$ such that
(i) $f(\R,n) \in \mathcal{O}(\log n)$ and (ii) $\eps_*(R_n, w, L_n) \leq 1/e$ 
for all input words $w \in \Sigma^*$ with $|w| \leq p(n)$.
\end{lemma}

\begin{proof}

Babu et al. \cite{BabuLRV13} have shown that for every language $K \in \mathbf{DLIN}$ there exists
a randomized streaming algorithm using space $\mathcal{O}(\log n)$ which, given an input $w$ of length $n$,
\begin{itemize}
\item accepts with probability 1 if $w \in K$,
\item and rejects with probability at least $1 - 1/n$ if $w \notin K$.
\end{itemize}
We remark that the algorithm needs to know the length of $w$ in advance.
To stay consistent with our definition,
 we view the algorithm above as a family $(S_n)_{n \ge 0}$ of randomized streaming algorithms $S_n$.
Furthermore, it is easy to see that the error probability $1/n$ can be further reduced to $1/n^d$ 
where $p(n) \le n^d$ for sufficiently large $n$ (by picking random primes of size $\Theta(n^{d+1})$ in the proof from \cite{BabuLRV13}).

Now we prove our claim for $L = \$ K_{\mathrm{pal}}$.
The streaming algorithm $R_n$ for window size $n$ works as follows:
After reading a $\$$-symbol, the algorithm $S_{n-1}$ from above is simulated
on the longest factor from $\{a,b\}^*$ that follows.
Simultaneously we maintain the length $\ell$ of the maximal suffix over $\{a,b\}$, up to $n$, using $\mathcal{O}(\log n)$ bits.
If $\ell$ reaches $n-1$, then  $R_n$ accepts if and only if $S_{n-1}$ accepts.
Notice that $R_n$ only errs if the stored length is $n-1$ (with probability $1/n^d$), which happens at most once
in every $n$ steps.
Therefore the number of time instants where $R_n$ errs on $w$ is bounded by $|w|/n \le n^d / n = n^{d-1}$.
By the union bound we have for every stream $w \in \{\$,a,b\}^{\le p(n)}$:
\[
	\eps_*(R_n,w,L_n) \le n^{d-1} \cdot \frac{1}{n^d} = \frac{1}{n}.
\]
This concludes the proof.
\end{proof}

\bibliographystyle{plain}
\bibliography{bib}

\end{document}